\documentclass[11pt]{article}

\RequirePackage{graphicx}
\RequirePackage{amsmath}
\RequirePackage{amsfonts}
\RequirePackage{amssymb}
\RequirePackage{amsthm}
\RequirePackage{mathtools}
\RequirePackage{bm}
\RequirePackage{booktabs}
\RequirePackage{microtype}
\RequirePackage{bbm}
\RequirePackage{algorithm}
\RequirePackage{algpseudocode}

\makeatletter
\@ifpackageloaded{natbib}{}{\RequirePackage{natbib}}
\@ifpackageloaded{hyperref}{}{\RequirePackage{hyperref}}
\makeatother
\setcitestyle{numbers}

\hypersetup{%
  pdftitle={Bayesian low-rank latent-cluster regression for mixed health outcomes},
  pdfauthor={Hsin-Hsiung Huang and Suyeon Kang},
  pdfsubject={Bayesian mixed-response latent-cluster regression},
  pdfkeywords={Bayesian clustering, low-rank regression, mixed responses, posterior similarity matrix}
}

\newcommand{\includeReadableFigure}[2][]{%
  \IfFileExists{#2.pdf}{\includegraphics[#1]{#2.pdf}}{\includegraphics[#1]{#2.png}}%
}

\makeatletter
\@ifundefined{frontmatter}{\newenvironment{frontmatter}{}{} }{}
\@ifundefined{runtitle}{\newcommand{\runtitle}[1]{}}{}
\@ifundefined{runauthor}{\newcommand{\runauthor}[1]{}}{}
\makeatother

\numberwithin{equation}{section}

\theoremstyle{plain}
\newtheorem{theorem}{Theorem}[section]
\newtheorem{lemma}[theorem]{Lemma}
\newtheorem{proposition}[theorem]{Proposition}

\theoremstyle{definition}

\theoremstyle{remark}
\newtheorem{remark}[theorem]{Remark}

\newcommand{\R}{\mathbb{R}}
\newcommand{\bzero}{\bm{0}}
\newcommand{\bI}{\bm{I}}
\newcommand{\bone}{\bm{1}}

\newcommand{\bmX}{\bm{X}}
\newcommand{\bmY}{\bm{Y}}

\newcommand{\bmu}{\bm{\mu}}
\newcommand{\bB}{\bm{B}}
\newcommand{\bL}{\bm{L}}
\newcommand{\bR}{\bm{R}}
\newcommand{\bU}{\bm{U}}
\newcommand{\bV}{\bm{V}}
\newcommand{\bD}{\bm{D}}
\newcommand{\bO}{\bm{O}}
\newcommand{\bP}{\bm{P}}
\newcommand{\bS}{\bm{S}}

\newcommand{\bGamma}{\bm{\Gamma}}

\newcommand{\bmeta}{\bm{\eta}}
\newcommand{\bpi}{\bm{\pi}}

\newcommand{\JG}{\mathcal{J}_{\mathrm{G}}}
\newcommand{\JB}{\mathcal{J}_{\mathrm{B}}}
\newcommand{\JNB}{\mathcal{J}_{\mathrm{NB}}}

\newcommand{\Normal}{\mathcal{N}}

\newcommand{\NegBin}{\mathrm{NegBin}}
\newcommand{\Dir}{\mathrm{Dir}}
\newcommand{\Ga}{\mathrm{Ga}}
\newcommand{\IG}{\mathrm{IG}}
\newcommand{\PG}{\mathrm{PG}}
\newcommand{\CRT}{\mathrm{CRT}}

\newcommand{\E}{\mathbb{E}}
\newcommand{\Prob}{\mathbb{P}}

\newcommand{\ind}{\mathbbm{1}}
\newcommand{\Ind}{\mathbbm{1}}

\newcommand{\WAIC}{\mathrm{WAIC}}

\newcommand{\rmax}{r_{\max}}

\newcommand{\logistic}{\operatorname{logit}^{-1}}

\DeclareMathOperator{\diag}{diag}
\DeclareMathOperator{\rank}{rank}
\DeclareMathOperator{\logit}{logit}
\DeclareMathOperator{\KL}{KL}

\DeclareMathOperator{\gap}{gap}

\makeatletter
\@ifundefined{kwd}{}{}
\makeatother

\begin{document}

\begin{frontmatter}

\runtitle{Mixed-type multivariate Bayesian clustering}
\runauthor{Huang and Kang}

\begingroup
\begin{center}
{\fontsize{16}{20}\selectfont\bfseries Bayesian low-rank latent-cluster regression for mixed health outcomes\par}
\vspace{0.8em}
{\normalsize Hsin-Hsiung Huang\textsuperscript{*} and Suyeon Kang\par}
\vspace{0.45em}
{\footnotesize School of Data, Mathematical, and Statistical Sciences, University of Central Florida,\par
Orlando, Florida 32826, USA\par}
\vspace{0.35em}
{\footnotesize \texttt{Hsin-Hsiung.Huang@ucf.edu}\par
\texttt{suyeon.kang@ucf.edu}\par}
\vspace{0.35em}
{\footnotesize \textsuperscript{*}Corresponding author: Hsin-Hsiung Huang, \texttt{Hsin-Hsiung.Huang@ucf.edu}.\par}
\vspace{1.0em}
{\bfseries Abstract}
\end{center}

\begin{quote}
\noindent
High-dimensional health and surveillance studies often involve many collinear predictors, multiple correlated outcomes of different types, and latent heterogeneity across observational units.
We propose a Bayesian latent-cluster reduced-rank regression model for multivariate mixed outcomes.
The model is a finite mixture of regression surfaces: each latent cluster has a cluster-specific mean shift and a low-rank coefficient matrix, yielding simultaneous clustering, dimension reduction, and component-wise interpretability.
Response coordinates may be Gaussian, Bernoulli, or negative binomial.
For the Bernoulli and negative binomial channels, P\'olya--Gamma augmentation provides a unified quadratic form that supports both a Gibbs sampler and a fast quadratic-surrogate coordinate-ascent variational routine for model screening.
Multiplicative gamma process shrinkage adapts the effective rank within each cluster, and a WAIC-based criterion is used to tune the number of clusters and the nominal maximal rank.
We establish posterior contraction for the identifiable component-specific regression surfaces and mean shifts, up to label permutation, and derive corresponding contraction for predictor-side singular subspaces.
We also analyze the default label-invariant reporting pipeline based on the posterior similarity matrix: an eigenspace embedding followed by mean shift is shown to consistently recover the latent partition under an additional strong separation margin.
Simulation experiments spanning all-Gaussian, all-Bernoulli, all-negative-binomial, and mixed Gaussian--Bernoulli--negative-binomial regimes show accurate recovery of the number of clusters and competitive clustering performance against $K$-means, \texttt{mclust}, PCA-based clustering, and a Gaussian reduced-rank mixture benchmark.
We illustrate the method in three applications: the \texttt{DoctorVisits} health-utilization study, Florida county COVID-19 surveillance, and U.S. state influenza surveillance.
The applications show how the model separates individual-level utilization groups and produces interpretable county- and state-level cluster maps together with response-specific posterior predictive maps.
\end{quote}



\endgroup

\end{frontmatter}

\section{Introduction}

Modern health and public health studies routinely record multiple outcomes together with large collections of potentially correlated covariates.
Examples range from utilization settings, where continuous self-reported health scores are observed alongside binary insurance indicators and discrete visit counts, to surveillance systems that track several endpoints per region, clinic, or patient.
These data present three recurring statistical challenges.
First, outcomes of different types, including continuous, binary, and count variables, are often jointly of interest and correlated, yet arise from distinct data-generating mechanisms.
Second, available covariates are commonly high dimensional and strongly collinear, reflecting overlapping sociodemographic, access-to-care, and reporting-intensity measures.
Third, the relationship between predictors and outcomes can vary across latent subpopulations, such as low- versus high-utilization groups or regions with different surveillance regimes.
Analyses that impose a single global regression surface can therefore obscure meaningful heterogeneity and yield effects that are difficult to interpret in practice.

A substantial literature addresses multivariate regression with mixed outcomes using generalized linear modeling ideas, with Bayesian formulations providing coherent uncertainty quantification and principled regularization.
A closely related contribution is the mixed-type multivariate Bayesian GLM of Wang, Bai and Huang \cite{WangBaiHuang2025}, which targets sparse variable selection in ultrahigh-dimensional settings through TPBN shrinkage priors.
That objective is distinct from the present low-rank latent-cluster formulation.
Here we target a complementary regime that is common in health and surveillance applications: predictors may be numerous and correlated, but the dominant structure is often shared across outcomes through a low-dimensional latent representation, and heterogeneity across observational units is better captured by allowing multiple regression surfaces rather than enforcing a single global one.
Accordingly, instead of emphasizing sparsity in a high-dimensional coefficient vector, we emphasize low-rank dimension reduction across outcomes and latent clustering across units, while retaining the ability to accommodate mixed response types within a single multivariate model.

Reduced-rank regression provides a principled approach to multivariate regression by projecting predictor effects onto a lower-dimensional latent structure, improving estimation efficiency and interpretability when the coefficient matrix is approximately low rank \citep{Anderson1951,ReinselVelu1998,Yuan2007,izenman2008}.
Bayesian low-rank formulations further adapt the effective dimension through shrinkage priors \citep{bhattacharya2011}.
However, classical reduced-rank models typically posit one regression surface shared across all units, which can be restrictive when predictor effects and cross-outcome relationships vary across latent subpopulations.
Finite mixture models provide a direct representation of such heterogeneity by allowing multiple regression surfaces indexed by latent components \citep{fraley2002model,bishop2006pattern}, but naive mixtures of multivariate regressions can be unstable when predictors are strongly correlated and outcomes are mixed-type.
These considerations motivate a framework that simultaneously clusters units by regression behavior and regularizes each cluster-specific multivariate regression surface through a parsimonious low-rank structure.

We develop a Bayesian latent-cluster low-rank regression model for multivariate mixed outcomes.
Within each latent cluster, the $p \times q$ coefficient matrix is modeled through a low-rank factorization, so that predictor effects across outcomes are driven by a small number of latent directions, while different clusters are permitted to exhibit different directions and magnitudes.
Response coordinates may be Gaussian, Bernoulli, or negative binomial within the same multivariate regression, a combination that aligns with common health utilization and surveillance endpoints.
To enable scalable computation across these likelihoods, we exploit P\'olya--Gamma augmentation \citep{polson2013bayesian} for Bernoulli and negative binomial components, yielding conditionally quadratic forms that support an efficient Gibbs sampler and also a fast quadratic-surrogate variational screening routine \citep{Blei2017}.
Rank adaptivity within each cluster is achieved through multiplicative gamma process shrinkage priors \citep{bhattacharya2011}, and we use Watanabe's information criterion for practical tuning of the number of clusters and the nominal maximal rank \citep{watanabe2010asymptotic}.

A central reporting issue in mixture models is label non-identifiability.
Rather than basing inference on arbitrary component labels, we report clustering through a label-invariant posterior similarity matrix (PSM).
For visualization and a stable default partition, we embed units in the leading eigenspace of the PSM and apply mean shift, a mode-seeking procedure that induces clusters via attraction basins of kernel density modes \citep{carreira2015meanshift,chen2016modeclustering}.
This strategy connects Bayesian posterior clustering uncertainty with a deterministic, label-invariant output and leverages the close relationship between eigenspace embeddings and spectral clustering ideas; related kernel mean-shift extensions in feature spaces are discussed by \citet{anand2014skms}.

Our contributions are threefold.
First, we propose a finite-mixture reduced-rank regression framework for multivariate mixed outcomes that combines cluster-specific mean shifts with cluster-specific low-rank coefficient matrices, thereby addressing correlated predictors, correlated outcomes, and latent heterogeneity within a single model.
Second, we develop a unified computational strategy based on P\'olya--Gamma augmentation, yielding both a Gibbs sampler for posterior inference and a fast variational screening procedure that supports WAIC-based tuning over $(K,\rmax)$.
Third, we establish posterior contraction for the identifiable cluster-specific regression surfaces, up to label permutation, and their associated singular subspaces, and we provide a consistency result for the default clustering pipeline based on PSM eigenspace embedding followed by mean shift under an additional strong separation margin.

We demonstrate the methodology in a model-comparison simulation study and in three real-data applications: the \texttt{DoctorVisits} health-utilization study, a Florida county COVID-19 analysis using hospitalization and death outcomes, and a U.S. state influenza analysis using weighted ILI rates and ILI patient counts.
The remainder of the paper is organized as follows.
Section~\ref{sec:related} positions the work relative to existing literature.
Section~\ref{sec:model} introduces the model and likelihood specification, Section~\ref{sec:comp} describes priors and computation, Section~\ref{sec:theory} presents the main theoretical results, and subsequent sections report simulation and data-analytic findings.

\section{Related Work}\label{sec:related}

The proposed framework is related to several strands of literature in multivariate regression, mixture modeling, and Bayesian low-rank inference.

Reduced-rank regression has a long history in multivariate analysis, beginning with classical formulations such as \citet{Anderson1951} and the broader treatment in \citet{ReinselVelu1998}, with modern statistical developments in \citet{Yuan2007,izenman2008}.
Penalized and structured extensions have incorporated low-rank regularization and sparsity to stabilize estimation in high dimensions; representative examples include matrix regularization and related structured multivariate regression methods \citep{Negahban2011,Witten2009}.
Bayesian low-rank formulations replace explicit rank constraints by shrinkage priors that adaptively learn the effective dimension, with the multiplicative gamma process being a particularly convenient device for factor-analytic and related models \citep{bhattacharya2011}.

Finite mixture models offer a natural representation of latent heterogeneity and provide a flexible foundation for model-based clustering \citep{fraley2002model,bishop2006pattern,scrucca2016mclust}.
Existing reduced-rank mixture regression approaches are developed largely for Gaussian outcomes and focus on component-specific low-rank coefficient matrices in settings where all response coordinates share the same likelihood family.
By contrast, the present paper addresses mixtures of multivariate regressions with genuinely mixed response types and develops a unified inferential pipeline that remains computationally tractable across Gaussian, Bernoulli, and negative binomial channels.

For mixed-response multivariate regression, the recent Bayesian literature has focused primarily on variable selection and structured sparsity.
Wang, Bai and Huang \cite{WangBaiHuang2025} give a representative example of that line of research in ultrahigh-dimensional settings, where the main goal is screening or selecting a relatively small set of important predictors.
The present paper instead targets settings in which latent clustering and shared low-dimensional outcome structure are more prominent than sparsity of the full coefficient array.

From the computational perspective, P\'olya--Gamma augmentation \citep{polson2013bayesian} has become a standard tool for logistic and related likelihoods, because it yields Gaussian conditional structure while preserving exact Bayesian updating.
Variational approximations provide a scalable alternative to MCMC in large model-selection problems \citep{Blei2017}, although they are typically best viewed as screening tools rather than replacements for full posterior sampling when uncertainty quantification is central.
Finally, posterior similarity matrices and related label-invariant summaries play an important role in Bayesian clustering \citep{rastelli2016optimal}, motivating the default PSM-based reporting strategy adopted here.

\section{A Bayesian Latent-Cluster Low-Rank Model for Mixed Responses}\label{sec:model}

Let $i=1,\dots,n$ index observational units, $j=1,\dots,q$ index responses, and $k=1,\dots,K$ index latent clusters.
We observe predictors $\bmX_i\in\R^p$ and a mixed-type response vector $\bmY_i=(Y_{i1},\dots,Y_{iq})^\top$.
Let $Z_i\in\{1,\dots,K\}$ denote the latent class with $\Prob(Z_i=k)=\pi_k$, where $\bpi=(\pi_1,\dots,\pi_K)$ and $\sum_{k=1}^K\pi_k=1$.

\subsection{Low-rank cluster-specific regression with mean shifts}\label{sec:lowrank}

Each cluster $k$ carries a $p\times q$ coefficient matrix $\bB_k$ factorized as
\[
\bB_k=\bL_k\bR_k^\top,\;
\bL_k\in\R^{p\times \rmax},\ \bR_k\in\R^{q\times \rmax},
\]
where shrinkage priors described in Section~\ref{sec:priors} adaptively learn an effective rank $r_k\le \rmax$.
We also include a cluster-specific mean shift, or intercept, $\bmu_k\in\R^q$.
The linear predictor for unit $i$ in cluster $k$ is
\begin{equation}\label{eq:eta-mu}
\bmeta_{ik}=\bmu_k+\bB_k^\top\bmX_i\in\R^q,\;
\eta_{ijk}=\mu_{kj}+\bmX_i^\top \bB_{k,\cdot j}.
\end{equation}

\paragraph{Identifiable cluster geometry and label-invariant reporting.}
The factorization $\bB_k=\bL_k\bR_k^\top$ is not unique because of latent rotations, and mixture labels are exchangeable because of label switching.
Therefore, our default clustering output is not based on component labels or on a particular rotation of $(\bL_k,\bR_k)$.
Instead, clustering and the default two-dimensional visualization are based on the label-invariant posterior similarity matrix defined in Section~\ref{sec:psm_meanshift}.
Cluster-specific regression interpretation is carried out on the identifiable scale of $\bB_k$; see Section~\ref{sec:interpret}.

\subsection{Mixed-likelihood specification}\label{sec:mixed}

Partition the response indices into disjoint sets $\JG$ for Gaussian responses, $\JB$ for Bernoulli responses, and $\JNB$ for negative binomial responses.
Conditional on $Z_i=k$, the Gaussian channel satisfies
\[
Y_{ij}\mid Z_i=k,\eta_{ijk},\sigma_j^2\sim \Normal(\eta_{ijk},\sigma_j^2),\; j\in\JG.
\]
The Bernoulli channel satisfies
\[
Y_{ij}\in\{0,1\},\; \Prob(Y_{ij}=1\mid Z_i=k)=\logistic(\eta_{ijk}),\; j\in\JB.
\]
For the negative binomial channel, we use the number-of-failures parameterization $\NegBin(r_j,p_{ijk})$ with
\[
\logit(p_{ijk})=\eta_{ijk},\;
\Prob(Y=y)=\binom{y+r_j-1}{y}p_{ijk}^y(1-p_{ijk})^{r_j},\; y=0,1,2,\dots,\; j\in\JNB.
\]
Equivalently, the conditional mean is
\[
\mu_{ijk}^{(NB)}:=\E[Y_{ij}\mid Z_i=k]
=
r_j\frac{p_{ijk}}{1-p_{ijk}}
=
r_j\exp(\eta_{ijk}).
\]

\begin{remark}[Offsets and exposures for counts]\label{rem:offset}
When counts $Y_{ij}$ are observed with exposure $E_{ij}>0$, one may include $\log E_{ij}$ as a fixed offset by replacing $\eta_{ijk}$ with $\eta_{ijk}+\log E_{ij}$.
All augmentation identities and Gaussian updates remain valid under this modification.
\end{remark}

\subsection{Unified quadratic form via P\'olya--Gamma augmentation}\label{sec:pg}

For $j\in\JB$, introduce latent variables $\omega_{ij}\sim \PG(1,\eta_{ijk})$ and define $\kappa_{ij}=Y_{ij}-\tfrac12$.
For $j\in\JNB$, introduce $\omega_{ij}\sim \PG(Y_{ij}+r_j,\eta_{ijk})$ and define
\[
\kappa_{ij}=\frac{Y_{ij}-r_j}{2}.
\]
Then, conditional on $\omega_{ij}$,
\begin{equation}\label{eq:pg-quad}
\log p(Y_{ij}\mid \eta_{ijk},\omega_{ij})
=\kappa_{ij}\eta_{ijk}-\tfrac12\omega_{ij}\eta_{ijk}^2+C(Y_{ij},r_j),
\end{equation}
for a term $C$ that does not depend on $\eta_{ijk}$.
For $j\in\JG$, the Gaussian likelihood has the same quadratic form by setting
$\omega_{ij}=\sigma_j^{-2}$ and $\kappa_{ij}=\sigma_j^{-2}Y_{ij}$, up to an additive constant.

Let $\mathcal{I}_k=\{i:Z_i=k\}$ and write $\eta_{k\cdot j}=(\eta_{ijk}:i\in\mathcal{I}_k)$,
$\kappa_{k j}=(\kappa_{ij}:i\in\mathcal{I}_k)$, and $\Omega_{k j}=\diag(\omega_{ij}:i\in\mathcal{I}_k)$.
Then the cluster-$k$ contribution to the augmented log-likelihood is
\begin{equation}\label{eq:cluster-quad}
\ell_k(\bmu_k,\bB_k)=
-\tfrac12\sum_{j=1}^q \eta_{k\cdot j}^\top \Omega_{k j}\eta_{k\cdot j}
+\sum_{j=1}^q \kappa_{k j}^\top \eta_{k\cdot j}
+\text{const}.
\end{equation}

\section{Priors, Posterior, Gibbs Sampler, Variational Inference, and PSM Clustering}\label{sec:comp}

\subsection{Priors and hyperpriors}\label{sec:priors}

\paragraph{Mixing proportions.}
We assume $\bpi\sim \Dir(\alpha_1,\dots,\alpha_K)$ with all $\alpha_k>0$.

\paragraph{Cluster-specific mean shifts.}
For $k=1,\dots,K$, we assume
\begin{equation}\label{eq:mu-prior}
\bmu_k\sim \Normal_q(\bzero,\sigma_\mu^2\bI_q),
\end{equation}
with fixed $\sigma_\mu^2$, recommended on standardized response scales, or alternatively $\sigma_\mu^{-2}\sim\Ga(a_\mu,b_\mu)$.

\paragraph{Gaussian variances.}
For $j\in\JG$, we assume $\sigma_j^2\sim \IG(a_\sigma,b_\sigma)$.

\paragraph{Negative binomial dispersions.}
For $j\in\JNB$, we assume $r_j\sim \Ga(a_r,b_r)$.

\paragraph{Multiplicative gamma process shrinkage for rank adaptivity.}
Write $\bL_k=[\ell_{k1},\dots,\ell_{k \rmax}]$ and $\bR_k=[r_{k1},\dots,r_{k \rmax}]$.
Let $\lambda_{kh}=\prod_{m=1}^h\delta_{km}$, with
\[
\phi_k\sim\Ga(a_\phi,b_\phi),\;
\delta_{k1}\sim \Ga(a_1,1),\;
\delta_{kh}\sim \Ga(a_2,1)\ (h\ge 2),\; a_2>a_1>1,
\]
and conditionally,
\[
\ell_{kh}\mid \phi_k,\lambda_{kh}\sim \Normal_p(\bzero,(\phi_k\lambda_{kh})^{-1}\bI_p),\;
r_{kh}\mid \phi_k,\lambda_{kh}\sim \Normal_q(\bzero,(\phi_k\lambda_{kh})^{-1}\bI_q),
\]
independently over $h$ and $k$.

\subsection{Posterior factorization and augmented representation}\label{sec:post}

Let $\theta$ collect $Z_{1:n}$, $\bpi$, $\{\bmu_k,\bB_k,\phi_k,\delta_k\}_{k=1}^K$,
$\{\sigma_j^2\}_{j\in\JG}$, $\{r_j\}_{j\in\JNB}$, and $\{\omega_{ij}\}$.
Up to a constant,
\[
\Pi(d\theta\mid \text{data})\propto
\Bigg[
\prod_{i=1}^n \pi_{Z_i}
\prod_{j\in\JG}\Normal\!\big(Y_{ij};\eta_{ij,Z_i},\sigma_j^2\big)
\prod_{j\in\JB} p(Y_{ij}\mid \eta_{ij,Z_i})
\prod_{j\in\JNB} p(Y_{ij}\mid \eta_{ij,Z_i},r_j)
\Bigg]\Pi(d\theta).
\]
After introducing P\'olya--Gamma latent variables for Bernoulli and negative binomial responses,
the augmented log-likelihood is jointly quadratic in the linear predictors, and hence quadratic in $(\bmu_k,\bB_k)$.

\subsection{Gibbs sampler (reference)}\label{sec:gibbs}

We record the Gibbs updates for completeness.
In the empirical studies below, we use variational inference for rapid screening and initialization, and we compute the default clustering from the posterior similarity matrix described in Section~\ref{sec:psm_meanshift}.

All formulas use $\eta_{ijk}=\mu_{kj}+\bmX_i^\top \bB_{k,\cdot j}$.

\paragraph{(1) Cluster labels $Z_i$.}
Let $\tilde w_{ik}\propto \pi_k \, p(\bmY_i\mid Z_i=k,\cdot)$.
Compute
\[
\begin{aligned}
\log \tilde w_{ik}={}&\log \pi_k
+\sum_{j\in\JG}\log \Normal(Y_{ij};\eta_{ijk},\sigma_j^2)
+\sum_{j\in\JB}\Big\{Y_{ij}\eta_{ijk}-\log(1+e^{\eta_{ijk}})\Big\}\\
&+\sum_{j\in\JNB}\Big\{Y_{ij}\eta_{ijk}-(Y_{ij}+r_j)\log(1+e^{\eta_{ijk}})\Big\}.
\end{aligned}
\]

\paragraph{(2) Mixing proportions.}
We update $\bpi\mid -\sim \Dir(\alpha_1+n_1,\dots,\alpha_K+n_K)$, where $n_k=\sum_i \ind(Z_i=k)$.

\paragraph{(3) P\'olya--Gamma latent variables.}
For $j\in\JB$, we update $\omega_{ij}\mid -\sim \PG(1,\eta_{ij,Z_i})$.
For $j\in\JNB$, we update $\omega_{ij}\mid -\sim \PG(Y_{ij}+r_j,\eta_{ij,Z_i})$.

\paragraph{(4) Gaussian variances.}
For $j\in\JG$,
\[
\sigma_j^2\mid -\sim \IG\!\Big(a_\sigma+\tfrac{n}{2},\ b_\sigma+\tfrac12\sum_{i=1}^n(Y_{ij}-\eta_{ij,Z_i})^2\Big).
\]

\paragraph{(5) Negative binomial dispersions.}
Using a CRT--Gamma representation, draw $L_{ij}\sim \CRT(Y_{ij},r_j)$, then
\[
r_j\mid -\sim \Ga\!\Big(a_r+\sum_i L_{ij},\ b_r+\sum_i \log(1+e^{\eta_{ij,Z_i}})\Big),\; j\in\JNB.
\]

\paragraph{(6) Mean shifts $\bmu_k$.}
Fix $k$ and $j$.
Let $\mathcal{I}_k=\{i:Z_i=k\}$, let $a_{ij}:=\bmX_i^\top\bB_{k,\cdot j}$, and let $\omega_{ij}$ and $\kappa_{ij}$ be as in \eqref{eq:pg-quad}; for Gaussian responses use $\omega_{ij}=\sigma_j^{-2}$ and $\kappa_{ij}=\sigma_j^{-2}Y_{ij}$.
Then
\[
\mu_{kj}\mid - \sim \Normal\big(m_{kj},v_{kj}\big),\;
v_{kj}=\Big(\sum_{i\in\mathcal{I}_k}\omega_{ij}+\sigma_\mu^{-2}\Big)^{-1},\;
m_{kj}=v_{kj}\sum_{i\in\mathcal{I}_k}(\kappa_{ij}-\omega_{ij}a_{ij}).
\]

\paragraph{(7) Updates for $\bR_k$.}
Let $X_k=(\bmX_i^\top)_{i\in\mathcal{I}_k}$ and $U_k=X_k\bL_k$.
Let $\Omega_{k j}=\diag(\omega_{ij}:i\in\mathcal{I}_k)$ and $\kappa_{k j}=(\kappa_{ij}:i\in\mathcal{I}_k)$.
Define $\tilde\kappa_{k j}=\kappa_{k j}-\Omega_{k j}\bone_{n_k}\mu_{kj}$.
With $\Lambda_k:=\phi_k\diag(\lambda_{k1},\dots,\lambda_{k \rmax})$,
\[
Q_{k j}=U_k^\top \Omega_{k j}U_k+\Lambda_k,\;
b_{k j}=U_k^\top \tilde\kappa_{k j},
\;
r_{k j}\mid -\sim \Normal_{\rmax}(Q_{k j}^{-1}b_{k j},\ Q_{k j}^{-1}).
\]

\paragraph{(8) Updates for $\bL_k$.}
Let $z_k=\mathrm{vec}(\bL_k)$.
Then
\[
Q^{(L)}_k=\sum_{j=1}^q (r_{k j}r_{k j}^\top)\otimes (X_k^\top \Omega_{k j}X_k)+\bI_p\otimes \Lambda_k,
\;
b^{(L)}_k=\sum_{j=1}^q (r_{k j}\otimes X_k^\top)\tilde\kappa_{k j},
\]
and
\[
\mathrm{vec}(\bL_k)\mid -\sim \Normal_{p\rmax}\big((Q^{(L)}_k)^{-1}b^{(L)}_k,\ (Q^{(L)}_k)^{-1}\big).
\]

\paragraph{(9) MGP hyperparameters.}
Let $S_{kh}=\|\ell_{kh}\|^2+\|r_{kh}\|^2$.
Then
\[
\phi_k\mid - \sim
\Ga\!\left(
 a_\phi+\tfrac{(p+q)\rmax}{2},\;
 b_\phi+\tfrac12\sum_{h=1}^{\rmax}\lambda_{kh}S_{kh}
\right),
\]
and $\{\delta_{kh}\}$ follow the standard MGP Gamma recursions.

\begin{algorithm}[t]
\caption{PG-augmented Gibbs sampler (reference) for mixed-response low-rank latent clusters}
\begin{algorithmic}[1]
\State Initialize $Z_{1:n}$, $\bpi$, $\{\bmu_k,\bL_k,\bR_k,\phi_k,\delta_k\}_{k=1}^K$, $\{\sigma_j^2\}_{j\in\JG}$, and $\{r_j\}_{j\in\JNB}$.
\For{$t=1$ to $T$}
  \State Update $Z_i$ for $i=1,\dots,n$.
  \State Update $\bpi$.
  \State Draw all $\omega_{ij}$.
  \State Update $\{\sigma_j^2\}_{j\in\JG}$ and $\{r_j\}_{j\in\JNB}$.
  \For{$k=1$ to $K$}
    \State Update $\bmu_k$.
    \State Update $\bR_k$ row by row.
    \State Update $\bL_k$.
    \State Update $(\phi_k,\delta_k)$.
  \EndFor
\EndFor
\end{algorithmic}
\label{alg:gibbs}
\end{algorithm}

\subsection{Variational inference with PG quadratic surrogates}\label{sec:vi}

We adopt a coordinate-ascent variational inference scheme that replaces the latent PG variables by their conditional expectation and performs weighted ridge-type updates for $(\bmu_k,\bL_k,\bR_k)$ together with soft assignments for $Z$.
Let $\gamma_{ik}:=q(Z_i=k)$ denote the responsibilities.

\paragraph{PG expectations.}
For $\omega\sim \PG(b,\eta)$,
\[
\E[\omega\mid \eta]=\frac{b}{2\eta}\tanh(\eta/2),
\;
\E[\omega\mid 0]=b/4.
\]

\paragraph{Pseudo-responses.}
For each response $j$ define
\[
z_{ij}=
\begin{cases}
Y_{ij}, & j\in\JG,\\
\kappa_{ij}/\omega_{ij}, & j\in\JB\cup\JNB,
\end{cases}
\;
\kappa_{ij}=
\begin{cases}
Y_{ij}-\tfrac12, & j\in\JB,\\
\tfrac{Y_{ij}-r_j}{2}, & j\in\JNB.
\end{cases}
\]
Let $w_{ikj}=\gamma_{ik}\omega_{ij}$ be the responsibility-weighted quadratic weights.
Define $W_{k j}=\diag(w_{1kj},\dots, w_{nkj})$ and $z_j=(z_{1j},\dots, z_{nj})^\top$.

\paragraph{Mean-shift update.}
Given current $\bB_k$, update $\bmu_k$ by weighted least squares.
For response $j$,
\begin{equation}\label{eq:vi-mu}
\mu_{kj}\leftarrow
\frac{\sum_{i=1}^n w_{ikj}\{z_{ij}-\bmX_i^\top\bB_{k,\cdot j}\}}
{\sum_{i=1}^n w_{ikj}+\sigma_\mu^{-2}}.
\end{equation}

\paragraph{R-step and L-step.}
Let $U_k=X\bL_k$ and define $\tilde z_{k j}=z_j-\mu_{kj}\bone_n$.
The R-step solves, for each $j$,
\[
Q_{k j}=(X\bL_k)^\top W_{k j}(X\bL_k)+\Lambda_k,\;
b_{k j}=(X\bL_k)^\top W_{k j}\tilde z_{k j},
\;
r_{k j}\leftarrow Q_{k j}^{-1}b_{k j}.
\]
The L-step solves the corresponding Kronecker normal equations with $\tilde z_{k j}$ and weights $W_{k j}$.
MGP hyperparameters are updated by their variational MAP forms using $\gamma$-weighted sufficient statistics.

\paragraph{Responsibilities and mixture weights.}
With current parameters, compute per-component log-likelihoods $\ell_{ik}$ using the true likelihood, not the quadratic surrogate, and update
\[
\gamma_{ik}\propto \pi_k\exp\{\ell_{ik}\},
\;
\pi_k\leftarrow
\frac{\alpha_k+\sum_{i=1}^n\gamma_{ik}}{\sum_{k'=1}^K\alpha_{k'}+\sum_{i=1}^n\sum_{k'=1}^K\gamma_{ik'}}.
\]

\begin{proposition}[Monotone ascent of PG-VI]\label{prop:vi-monotone}
Fix the hyperparameters.
One full outer iteration of the coordinate-ascent scheme produces a nondecreasing sequence of ELBO values.
Any limit point is a coordinate-wise maximizer of the ELBO and hence a stationary point of the profiled ELBO.
\end{proposition}

\subsection{PSM-based clustering and default two-dimensional visualization}\label{sec:psm_meanshift}

Mixture labels are not identifiable, so we report clustering through the label-invariant posterior similarity matrix
\[
\bS_{ii'} := \Pi(Z_i=Z_{i'}\mid \text{data}) \in [0,1].
\]
In full posterior sampling,
\[
\widehat{S}_{ii'}=\frac{1}{T}\sum_{t=1}^T \Ind\{Z_i^{(t)}=Z_{i'}^{(t)}\}.
\]
Under variational inference, using conditional independence of labels given parameters,
\[
\bS_{ii'} \approx \sum_{k=1}^K \gamma_{ik}\gamma_{i'k}.
\]
Thus we use the deterministic approximation
\begin{equation}\label{eq:psm_vi}
\widehat{\bS} \approx \widehat{\bGamma}\widehat{\bGamma}^\top,
\;
\widehat{\bGamma}=(\widehat{\gamma}_{ik})\in\R^{n\times K}.
\end{equation}

\paragraph{Efficient eigendecomposition under variational inference.}
Under variational inference, $\widehat{\bS}=\widehat{\bGamma}\widehat{\bGamma}^\top$ has rank at most $K$.
Thus the leading eigenspace of $\widehat{\bS}$ can be obtained without forming the $n\times n$ matrix.
Let $\widehat{\bGamma}^\top\widehat{\bGamma}=\bV\bD\bV^\top$ be the eigendecomposition of the $K\times K$ matrix, with $\bD=\diag(d_1,\ldots,d_K)$.
For each $d_k>0$, the corresponding eigenvector of $\widehat{\bS}$ is proportional to $\widehat{\bGamma}\,\bV_{\cdot k}/\sqrt{d_k}$, so we may take
$\widehat{\bU}=\widehat{\bGamma}\bV\bD^{-1/2}$ after dropping any zero eigenvalues.
This reduces the eigendecomposition cost from $O(n^3)$ to $O(nK^2+K^3)$.

Let $\widehat{\bU}\in\R^{n\times K}$ denote the top-$K$ eigenvectors of $\widehat{\bS}$, and let $\widehat{u}_i^\top$ be the $i$th row of $\widehat{\bU}$.
We run mean shift on $\{\widehat{u}_i\}$ and obtain a partition by attraction basins of kernel-density modes.
The default two-dimensional plot is a scatter of $(\widehat{u}_{i1},\widehat{u}_{i2})$ colored by the resulting mean-shift cluster labels.

\begin{algorithm}[t]
\caption{Default clustering and two-dimensional visualization: PSM eigenspace plus mean shift}
\begin{algorithmic}[1]
\State Fit the model and obtain responsibilities $\widehat{\bGamma}=(\widehat{\gamma}_{ik})$ or MCMC draws of $Z$.
\State Form $\widehat{\bS}=\widehat{\bGamma}\widehat{\bGamma}^\top$ or $\widehat{S}_{ii'}=\tfrac{1}{T}\sum_t \Ind\{Z_i^{(t)}=Z_{i'}^{(t)}\}$.
\State Compute the top-$K$ eigenvectors $\widehat{\bU}=[\widehat{u}_{\cdot 1},\dots,\widehat{u}_{\cdot K}]$ of $\widehat{\bS}$.
\State Embed unit $i$ by $\widehat{u}_i^\top$, the $i$th row of $\widehat{\bU}$, and optionally row-normalize.
\State Run mean shift in $\R^K$ on $\{\widehat{u}_i\}$ with bandwidth $h$ to obtain clusters.
\State Plot $(\widehat{u}_{i1},\widehat{u}_{i2})$ colored by the mean-shift clusters.
\end{algorithmic}
\label{alg:psm_ms}
\end{algorithm}

\subsection{Model selection and hyperparameter tuning via WAIC}\label{sec:waic}

For variational fits, define $\ell_{ik}=\log p(\bm y_i\mid Z_i=k,\bm x_i,\hat\theta_k)$.
The plug-in lppd is
\[
\mathrm{lppd}_{\mathrm{VI}}=\sum_{i=1}^n \log\Big(\sum_{k=1}^K \hat\pi_k\,e^{\ell_{ik}}\Big).
\]
We use a responsibility-weighted surrogate variance
\[
\bar\ell_i=\sum_{k=1}^K \widehat\gamma_{ik}\,\ell_{ik},
\;
p_{\mathrm{WAIC,VI}}=\sum_{i=1}^n\sum_{k=1}^K \widehat\gamma_{ik}\,(\ell_{ik}-\bar\ell_i)^2,
\]
and report
\[
\mathrm{WAIC}_{\mathrm{VI}}=-2\big(\mathrm{lppd}_{\mathrm{VI}}-p_{\mathrm{WAIC,VI}}\big).
\]
We treat $\mathrm{WAIC}_{\mathrm{VI}}$ primarily as a computational screening criterion for ranking candidate specifications, and we apply a one-standard-error rule together with a minimum-mixture safeguard to avoid nearly empty clusters.

\subsection{Rotation-invariant interpretation via eigenspaces}\label{sec:interpret}

The factorization $\bB_k=\bL_k\bR_k^\top$ is rotation non-unique, so we interpret coefficients on the identifiable scale of $\bB_k$ and its singular value decomposition.
Let $\bB_k=\bU_k\bD_k\bV_k^\top$.
Predictor-side and response-side energies are
\[
I^{(X)}_{k,a}=\|\bB_{k,a\cdot}\|_2^2=\sum_{h=1}^{\rmax} d_{kh}^2\,\bU_k[a,h]^2,\;
I^{(Y)}_{k,j}=\|\bB_{k,\cdot j}\|_2^2=\sum_{h=1}^{\rmax} d_{kh}^2\,\bV_k[j,h]^2.
\]
Cluster summaries are therefore reported as $(\bmu_k,\bB_k)$, aligned by permutation, and clustering itself is reported via the PSM.

\section{Theoretical Properties}\label{sec:theory}

We state posterior contraction and consistency results for the identifiable parameters $(\bmu_{1:K},\bB_{1:K})$ under fixed $(K,\rmax)$, and we propagate these results to predictor-side eigenspaces.
Complete proofs are provided in Appendix~\ref{app:proofs}.

\subsection{Targets and metrics}\label{sec:targets}

Write the conditional density for $\bmY$ given $\bmX=x$ under $\theta$ as
\[
f_\theta(\bm y\mid x)=\sum_{k=1}^K \pi_k \prod_{j=1}^q f_{j}\bigl(y_j;\eta_{kj}(x),\psi_j\bigr),
\;
\eta_k(x)=\bmu_k+\bB_k^\top x,
\]
where $f_j$ is Gaussian for $j\in\JG$, Bernoulli-logit for $j\in\JB$, and negative binomial with logit-$p$ parameterization for $j\in\JNB$.
Let $P_X$ denote the marginal law of $\bmX$ and let $P_\theta$ denote the joint law of $(\bmX,\bmY)$.

\paragraph{Integrated Hellinger distance.}
Define
\[
h^2\!\bigl(f_\theta(\cdot\mid x),f_{\theta_0}(\cdot\mid x)\bigr)
=\frac12\int \Big(\sqrt{f_\theta(\bm y\mid x)}-\sqrt{f_{\theta_0}(\bm y\mid x)}\Big)^2\,d\bm y,
\;
d_H^2(\theta,\theta_0)=\int h^2(\cdot)\,P_X(dx).
\]

\paragraph{Parameter distance up to permutation.}
For a permutation $\tau$ define
\[
\|\theta-\theta_0\|_\tau^2
=\|\bpi-\bpi_0^\tau\|_1^2+\sum_{k=1}^K\|\bmu_k-\bmu_{0,\tau(k)}\|_2^2
+\sum_{k=1}^K\|\bB_k-\bB_{0,\tau(k)}\|_F^2
+\sum_{j\in\JG}|\sigma_j^2-\sigma_{0j}^2|^2+\sum_{j\in\JNB}|r_j-r_{0j}|^2,
\]
and define $d_\Theta(\theta,\theta_0)=\min_\tau \|\theta-\theta_0\|_\tau$.

\paragraph{Eigenspace targets.}
Let $\bB_k=\bU_k\bD_k\bV_k^\top$ and define the predictor-side projector
\[
\bP_k=\bU_{k,r_{0k}}\bU_{k,r_{0k}}^\top,
\;
d_{\mathcal{S}}^2(\theta,\theta_0)=\min_{\tau}\sum_{k=1}^K \|\bP_k-\bP_{0,\tau(k)}\|_F^2.
\]

\paragraph{PSM target.}
Define the posterior similarity matrix $\bS$ by $\bS_{ii'}=\Pi(Z_i=Z_{i'}\mid \text{data})$.
Under an additional strong separation condition, Assumption (B8') below, $\bS$ concentrates near the hard similarity matrix $\bS_0$,
\[
(\bS_0)_{ii'}=\ind(Z_{0i}=Z_{0i'}).
\]

\subsection{Assumptions}\label{sec:assumptions_new}

Let $(\bmX_i,\bmY_i)_{i=1}^n$ be i.i.d.\ from $P_{\theta_0}$ with
\[
\theta_0=\{\bpi_0,(\bmu_{0k},\bB_{0k})_{k=1}^K,(\sigma_{0j}^2)_{j\in\JG},(r_{0j})_{j\in\JNB}\}.
\]

\paragraph{(B1) Design regularity.}
The predictor vector $\bmX$ has compact support and $\E[\bmX\bmX^\top]$ is positive definite.

\paragraph{(B2) True low-rank structure.}
Each $\bB_{0k}$ has rank $r_{0k}\le \rmax$, with $\|\bB_{0k}\|_F<\infty$ and $\|\bmu_{0k}\|_2<\infty$.

\paragraph{(B3) Component separation and nondegeneracy.}
There exists $\pi_{\min}>0$ such that $\min_k \pi_{0k}\ge \pi_{\min}$.
There exists $\delta_{\mathrm{sep}}>0$ such that for all $k\neq \ell$,
\[
\int h^2\!\bigl(f_{0k}(\cdot\mid x),f_{0\ell}(\cdot\mid x)\bigr)\,P_X(dx)\ge \delta_{\mathrm{sep}},
\]
where $f_{0k}(\cdot\mid x)=\prod_j f_j(\cdot;\eta_{0,kj}(x),\psi_{0j})$.

\paragraph{(B4) Regular likelihoods on compacts.}
For $j\in\JG$, $\sigma_{0j}^2\in(0,\infty)$.
For $j\in\JB\cup\JNB$, the log-likelihood in $\eta$ is three-times continuously differentiable, and on any compact interval $I\subset\R$ the conditional Fisher information in $\eta$ is bounded above and below away from $0$.

\paragraph{(B5) Identifiability up to permutation.}
If $f_\theta(\cdot\mid x)=f_{\theta'}(\cdot\mid x)$ for $P_X$-almost every $x$, then $\theta$ and $\theta'$ coincide up to a permutation of component labels.

\paragraph{(B6) Prior thickness and sieve tails.}
The prior assigns positive mass to every neighborhood of $\theta_0$ in $d_\Theta$.
Moreover, there exists a sieve $\Theta_n$ with $\Pi(\Theta_n^c)\le e^{-c n\varepsilon_n^2}$ and
$\log N(\varepsilon_n,\Theta_n,d_H)\lesssim n\varepsilon_n^2$ for the rate $\varepsilon_n$ defined below.

\paragraph{(B7) Eigenspace gap.}
For each $k$,
\[
g_k := s_{r_{0k}}(\bB_{0k})-s_{r_{0k}+1}(\bB_{0k})>0,
\]
with $s_{r_{0k}+1}=0$ if $\rank(\bB_{0k})=r_{0k}$.

\paragraph{(B8') Perfect-separation margin for exact partition recovery via the PSM.}
There exists $\Delta>0$ such that under $\theta_0$, almost surely,
\[
\log\frac{\pi_{0,Z_i}\,f_{0,Z_i}(\bmY_i\mid \bmX_i)}{
\max_{\ell\neq Z_i}\pi_{0\ell}\,f_{0\ell}(\bmY_i\mid \bmX_i)}
\ge \Delta
\;\text{for all } i.
\]
This implies zero Bayes classification error and yields PSM consistency for the hard similarity matrix $\bS_0$.

\begin{remark}[Why Assumption (B8') is separated out]
Posterior contraction for $(\bmu_k,\bB_k)$ does not require Assumption (B8').
Assumption (B8') is used only to prove exact recovery of the latent partition from the PSM followed by mean shift.
Without (B8'), the model parameters can be consistent while a nonzero classification error persists.
\end{remark}

\subsection{Posterior contraction and eigenspace consistency}\label{sec:post_contr}

Define the effective dimension
\[
d_n = (K-1) + Kq + K \rmax(p+q) + |\JG| + |\JNB|.
\]
Define the nominal rate
\[
\varepsilon_n=\sqrt{\frac{d_n\log n}{n}}.
\]

\begin{theorem}[Posterior contraction in integrated Hellinger]\label{thm:post_contr_H}
Assume (B1), (B2), (B4), (B5), and (B6) with fixed $(K,\rmax)$ and $d_n=o(n/\log n)$.
Then there exists $M>0$ such that
\[
\Pi\bigl(d_H(\theta,\theta_0)>M\varepsilon_n \,\big|\, (\bmX_i,\bmY_i)_{i=1}^n\bigr)\to 0
\;\text{in }P_{\theta_0}\text{-probability}.
\]
\end{theorem}

\begin{theorem}[Parameter contraction]\label{thm:post_contr_param}
Assume (B1) through (B6) and the separation condition (B3).
Then there exists $M'>0$ such that
\[
\Pi\bigl(d_\Theta(\theta,\theta_0)>M'\varepsilon_n \,\big|\, \text{data}\bigr)\to 0
\;\text{in }P_{\theta_0}\text{-probability}.
\]
In particular, up to label permutation,
\[
\max_k \|\bmu_k-\bmu_{0k}\|_2 = O_{P_{\theta_0}}(\varepsilon_n)
\;\text{and}\;
\max_k \|\bB_k-\bB_{0k}\|_F = O_{P_{\theta_0}}(\varepsilon_n).
\]
\end{theorem}

\begin{theorem}[Posterior contraction for predictor-side eigenspaces]\label{thm:eigenspace_contr}
Assume (B1) through (B7).
Then there exists $C>0$ such that
\[
\Pi\Bigl(d_{\mathcal{S}}(\theta,\theta_0) > C\varepsilon_n \,\Big|\,\text{data}\Bigr)\to 0
\;\text{in }P_{\theta_0}\text{-probability}.
\]
Moreover, for each $k$ after label alignment,
\[
\|\bP_k-\bP_{0k}\|_F
\le \frac{4}{g_k}\,\|\bB_k-\bB_{0k}\|_F
\]
holds on an event whose posterior probability tends to $1$.
\end{theorem}

\subsection{Consistency of PSM eigenspace plus mean shift clustering}\label{sec:psm_ms_theory}

Let $\bS_{ii'}=\Pi(Z_i=Z_{i'}\mid \text{data})$ be the posterior similarity matrix.
Let $\hat\bU$ be the top-$K$ eigenvectors of $\bS$ and $\hat u_i^\top$ the $i$th row.
Run mean shift with bandwidth $h=h_n$ on $\{\hat u_i\}\subset\R^K$ to obtain a partition $\hat{\mathcal{C}}_n$.

\begin{theorem}[Consistency of PSM eigenspace plus mean shift]\label{thm:psm_meanshift_consistency}
Assume (B1) through (B7) and the perfect-separation margin (B8').
Let $\hat{\mathcal{C}}_n$ be the partition produced by the top-$K$ eigenspace of $\bS$ followed by mean shift with bandwidth $h_n$.
Suppose $h_n\to 0$ and $\varepsilon_n/h_n\to 0$.
Then, up to a permutation of labels,
\[
\Prob_{\theta_0}\!\left(
\frac{1}{n}\min_{\tau}\sum_{i=1}^n
\ind\{\hat{\mathcal{C}}_n(i)\neq \tau(Z_{0i})\}>0
\right)\to 0.
\]
\end{theorem}

\paragraph{Implementation note for the variational PSM.}
When using variational inference, we replace $\bS$ by $\widehat{\bS}=\widehat{\bGamma}\widehat{\bGamma}^\top$ as in \eqref{eq:psm_vi}.
Appendix~\ref{app:proofs} records a deterministic perturbation lemma: if $\|\widehat{\bS}-\bS\|_F=o_P(h_n)$ and the PSM eigengap at rank $K$ is of order $n$, then the same conclusion holds with $\widehat{\bS}$ in place of $\bS$.

\section{Simulation Study}\label{sec:sim}

We conducted a simulation study to evaluate clustering accuracy, model-selection behavior, and predictive performance across response-family configurations.
The simulation suite was generated from the model class in Section~\ref{sec:model}, with two latent clusters, three responses, true rank two, and the same low-rank coefficient construction across scenarios.
The computational implementation uses variational inference and the PSM-based clustering rule as the default model-comparison pipeline.
It also compares the proposed method with generic clustering baselines and, in the all-Gaussian setting, with a reduced-rank mixture benchmark implemented through \texttt{rrMixture} \citep{rrMixture2022}.

\subsection{Scenario design}\label{sec:sim-design-0506}

Table~\ref{tab:sim_scenarios_0506} summarizes the four response-family scenarios.
Each scenario uses $n=1000$, $p=40$, $K_{\mathrm{true}}=2$, and true rank $r_{\mathrm{true}}=2$.
Cluster separation is controlled by mean-shift and coefficient-surface separation parameters $(\mathrm{sep}_{\mu},\mathrm{sep}_{B})=(3.5,7.5)$.
For Gaussian responses the noise standard deviation is $0.7$; for negative binomial responses the size parameter is $r_{\mathrm{NB}}=12$.
Predictors are generated from a mean-zero multivariate normal distribution with AR(1) covariance, and the response family is varied across the four scenarios.

\begin{table}[t]
\centering
\caption{Simulation scenarios used in the model-comparison study. All scenarios use $n=1000$, $p=40$, $K_{\mathrm{true}}=2$, true rank $r_{\mathrm{true}}=2$, $\mathrm{sep}_{\mu}=3.5$, $\mathrm{sep}_{B}=7.5$, Gaussian standard deviation $0.7$, and negative binomial size $r_{\mathrm{NB}}=12$ when applicable.}
\label{tab:sim_scenarios_0506}
\resizebox{\linewidth}{!}{%
\begin{tabular}{rlllrr}
\toprule
ID & Scenario & Response profile & rrMix benchmark & $n$ & $p$ \\
\midrule
1 & All Gaussian & Gaussian / Gaussian / Gaussian & Yes & 1000 & 40 \\
2 & All binary & Bernoulli / Bernoulli / Bernoulli & No & 1000 & 40 \\
3 & All count & NB / NB / NB & No & 1000 & 40 \\
4 & One of each & Gaussian / Bernoulli / NB & No & 1000 & 40 \\
\bottomrule
\end{tabular}
}%
\end{table}

\subsection{Methods compared}\label{sec:sim-methods-0506}

The proposed procedure, denoted BMLC-VI-PSM, fits the mixed-response low-rank latent-cluster model by the variational routine in Section~\ref{sec:vi}, selects the working model over a small grid of $(K,\rmax)$, and reports clustering through the PSM eigenspace followed by mean shift.
The comparison set contains $K$-means on responses only, $K$-means on concatenated predictors and responses, Gaussian finite-mixture model-based clustering via \texttt{mclust} on responses only and on concatenated predictors and responses, and PCA followed by $K$-means on the concatenated feature matrix.
For mixed-response feature representations, Gaussian coordinates are centered and scaled, Bernoulli coordinates are kept as $0/1$ indicators, and count coordinates are transformed by $\log(1+y)$ before standardization.
For the all-Gaussian scenario only, we also report \texttt{rrMix}, a Gaussian reduced-rank mixture benchmark.
The number of clusters is fixed at the true value $K=2$ for the generic clustering baselines, so that the comparison targets clustering quality rather than unsupervised model-order selection \citep{scrucca2016mclust}.

\subsection{Clustering performance}\label{sec:sim-clustering-0506}

Table~\ref{tab:sim_cluster_compare_0506} reports the main clustering comparison over 100 replications per scenario.
The proposed BMLC-VI-PSM method attains mean accuracy $0.996$ in the all-Gaussian scenario, $0.974$ in the all-count scenario, and $0.954$ in the mixed Gaussian--Bernoulli--negative-binomial scenario.
In the all-Gaussian scenario it essentially matches \texttt{rrMix} in accuracy, with a slightly smaller mean ARI.
In the all-binary scenario, BMLC-VI-PSM remains accurate but is modestly below the response-only $K$-means baseline, indicating that the likelihood-based low-rank mixture is not uniformly dominant when all coordinates are binary and the response-space separation is already strong.

\begin{table}[t]
\centering
\caption{Simulation clustering comparison. Values are means over 100 replications; parentheses show standard deviations for accuracy and ARI. The generic baselines use the true number of clusters, $K=2$, so the comparison isolates clustering accuracy rather than model-order selection. HCE denotes hard classification error after label alignment.}
\label{tab:sim_cluster_compare_0506}
\scriptsize
\resizebox{\linewidth}{!}{%
\begin{tabular}{llcccc}
\toprule
Scenario & Method & Accuracy & ARI & Macro F1 & HCE \\
\midrule
All Gaussian & BMLC-VI-PSM & 0.996 (0.004) & 0.983 (0.015) & 0.996 & 0.004 \\
All Gaussian & rrMix & 0.996 (0.003) & 0.986 (0.012) & 0.996 & 0.004 \\
All Gaussian & KMeans(Y) & 0.974 (0.027) & 0.900 (0.096) & 0.974 & 0.026 \\
All Gaussian & KMeans(X,Y) & 0.933 (0.080) & 0.776 (0.232) & 0.933 & 0.067 \\
All Gaussian & Mclust(Y) & 0.988 (0.012) & 0.953 (0.047) & 0.988 & 0.012 \\
All Gaussian & Mclust(X,Y) & 0.513 (0.009) & 0.000 (0.000) & 0.340 & 0.487 \\
All Gaussian & PCA-KMeans(X,Y) & 0.873 (0.084) & 0.585 (0.228) & 0.873 & 0.127 \\
\midrule
All binary & BMLC-VI-PSM & 0.917 (0.047) & 0.703 (0.146) & 0.917 & 0.083 \\
All binary & KMeans(Y) & 0.930 (0.035) & 0.745 (0.110) & 0.930 & 0.070 \\
All binary & KMeans(X,Y) & 0.899 (0.068) & 0.655 (0.196) & 0.899 & 0.101 \\
All binary & Mclust(Y) & 0.911 (0.036) & 0.680 (0.117) & 0.911 & 0.089 \\
All binary & Mclust(X,Y) & 0.513 (0.009) & 0.000 (0.000) & 0.340 & 0.487 \\
All binary & PCA-KMeans(X,Y) & 0.846 (0.075) & 0.500 (0.191) & 0.846 & 0.154 \\
\midrule
All count & BMLC-VI-PSM & 0.974 (0.025) & 0.899 (0.089) & 0.973 & 0.026 \\
All count & KMeans(Y) & 0.969 (0.033) & 0.883 (0.114) & 0.969 & 0.031 \\
All count & KMeans(X,Y) & 0.947 (0.056) & 0.810 (0.176) & 0.947 & 0.053 \\
All count & Mclust(Y) & 0.961 (0.068) & 0.868 (0.209) & 0.960 & 0.039 \\
All count & Mclust(X,Y) & 0.513 (0.009) & 0.000 (0.000) & 0.340 & 0.487 \\
All count & PCA-KMeans(X,Y) & 0.920 (0.063) & 0.720 (0.190) & 0.919 & 0.080 \\
\midrule
One Gaussian, one binary, one count & BMLC-VI-PSM & 0.954 (0.054) & 0.836 (0.175) & 0.954 & 0.046 \\
One Gaussian, one binary, one count & KMeans(Y) & 0.944 (0.040) & 0.796 (0.131) & 0.944 & 0.056 \\
One Gaussian, one binary, one count & KMeans(X,Y) & 0.911 (0.063) & 0.692 (0.189) & 0.911 & 0.089 \\
One Gaussian, one binary, one count & Mclust(Y) & 0.909 (0.044) & 0.675 (0.136) & 0.908 & 0.091 \\
One Gaussian, one binary, one count & Mclust(X,Y) & 0.513 (0.009) & 0.000 (0.000) & 0.340 & 0.487 \\
One Gaussian, one binary, one count & PCA-KMeans(X,Y) & 0.879 (0.071) & 0.594 (0.197) & 0.879 & 0.121 \\
\bottomrule
\end{tabular}
}%
\end{table}

Figures~\ref{fig:sim_accuracy_0506}--\ref{fig:sim_bmlc_accuracy_0506} summarize the same comparison graphically.
The dot-and-interval panels show accuracy and ARI by scenario, while the heatmap gives a compact view of the mean accuracy surface across methods, including the \texttt{mclust} baselines.
Each plot is displayed at full text width to preserve axis, legend, and annotation readability.
The proposed method is strongest or tied in the all-Gaussian, all-count, and mixed-response settings. In the all-binary scenario, response-only $K$-means and \texttt{mclust} are competitive, which is expected when the marginal binary response space is already well separated.

\begin{figure}[t]
\centering
\includeReadableFigure[width=0.98\linewidth]{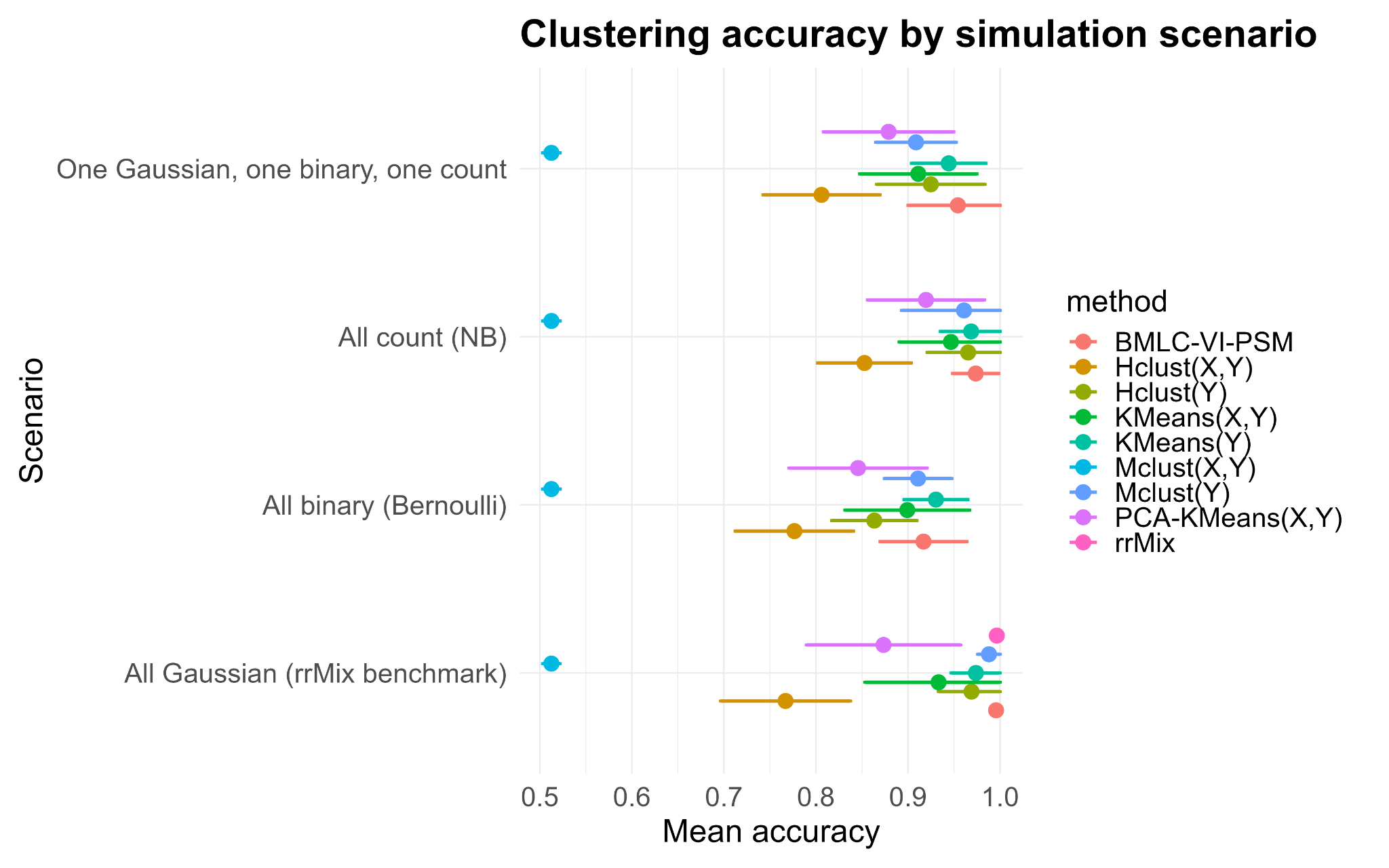}
\caption{Simulation comparison by scenario: mean clustering accuracy with standard-deviation interval summaries over 100 replications. The figure is exported with enlarged fonts and displayed at full text width for readability.}
\label{fig:sim_accuracy_0506}
\end{figure}

\begin{figure}[t]
\centering
\includeReadableFigure[width=0.98\linewidth]{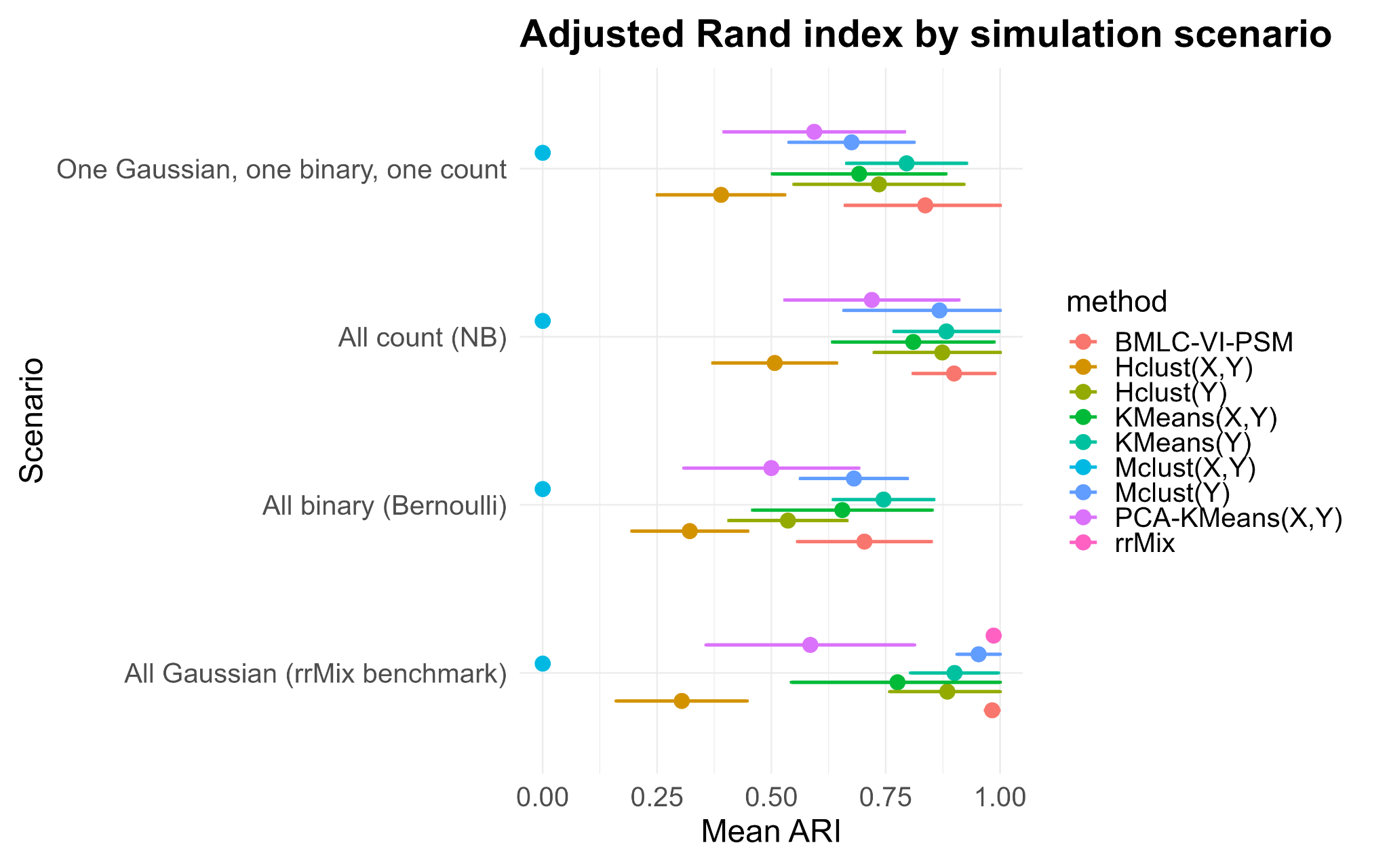}
\caption{Simulation comparison by scenario: mean adjusted Rand index with standard-deviation interval summaries over 100 replications.}
\label{fig:sim_ari_0506}
\end{figure}

\begin{figure}[t]
\centering
\includeReadableFigure[width=0.98\linewidth]{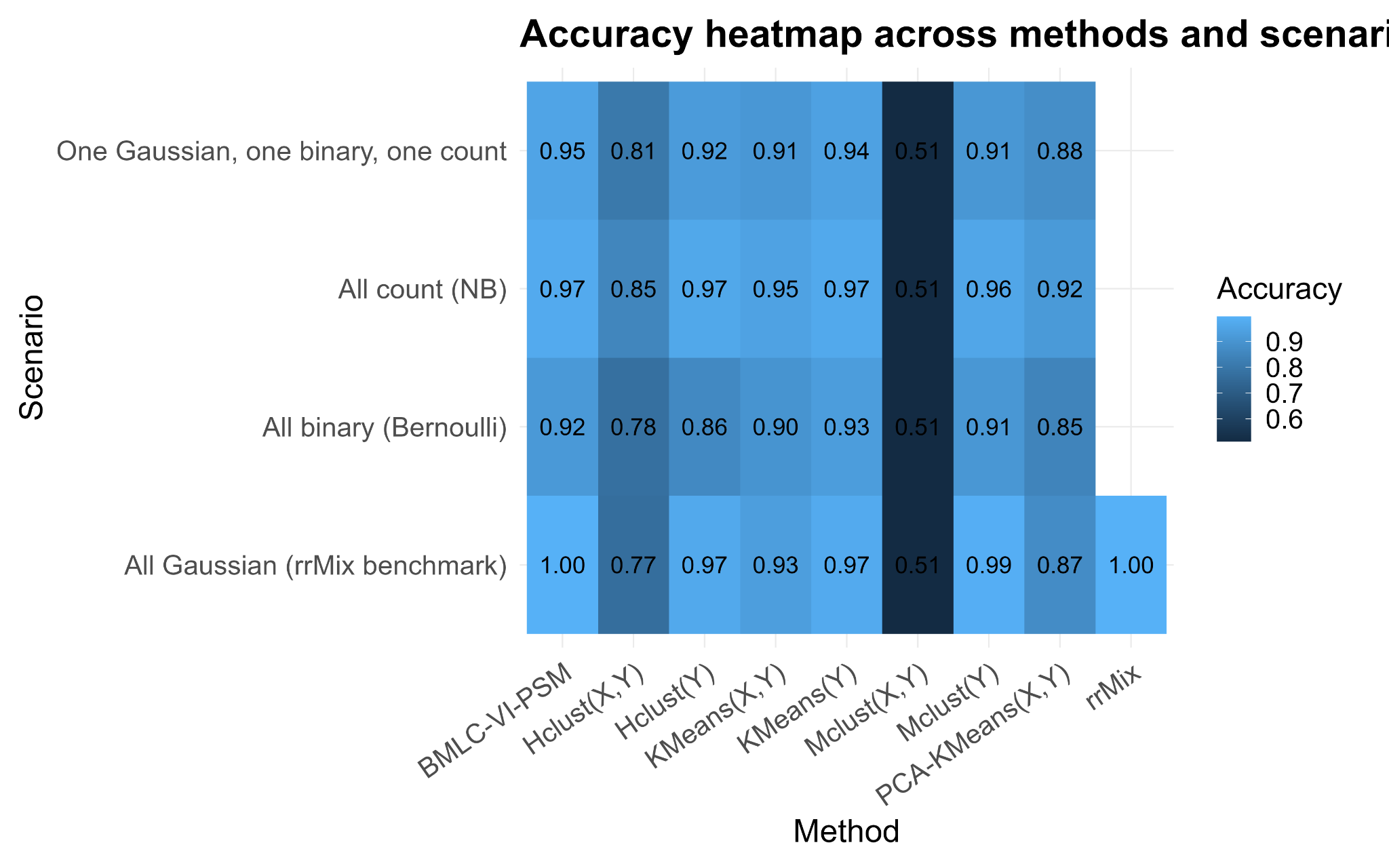}
\caption{Accuracy heatmap for the simulation suite. Cells show mean clustering accuracy by method and scenario.}
\label{fig:sim_heatmap_0506}
\end{figure}

\begin{figure}[t]
\centering
\includeReadableFigure[width=0.94\linewidth]{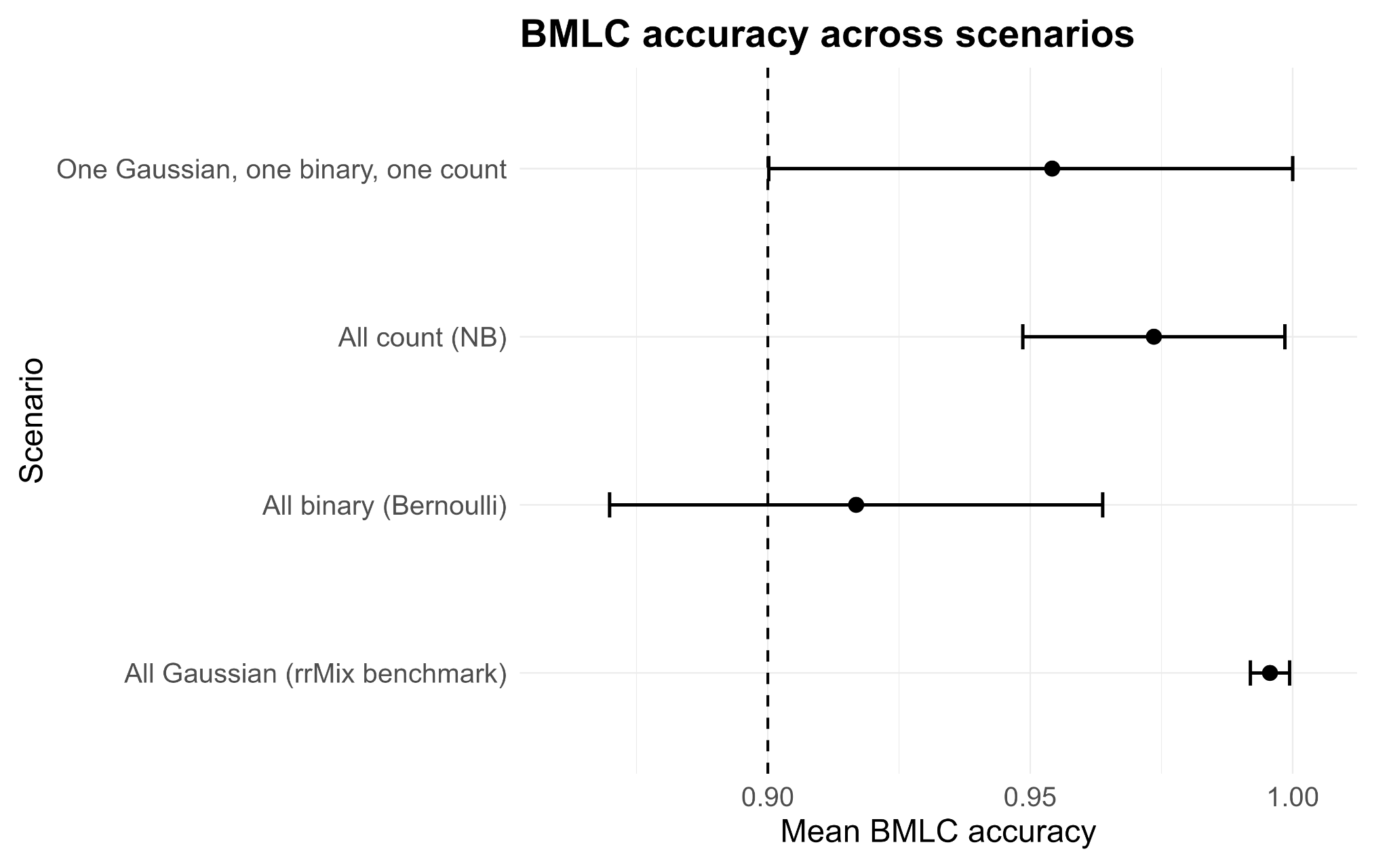}
\caption{BMLC-VI-PSM accuracy by simulation scenario. The dashed vertical line marks the 0.90 target accuracy threshold; the threshold is identified only in the caption to avoid label overlap inside the panel.}
\label{fig:sim_bmlc_accuracy_0506}
\end{figure}

\subsection{Model recovery and predictive summaries}\label{sec:sim-recovery-0506}

Table~\ref{tab:sim_bmlc_summary_0506} summarizes BMLC-specific recovery and predictive metrics.
The variational tuning procedure recovers $K=2$ in all four scenarios.
The selected rank is exactly two in the all-Gaussian, all-binary, and all-count scenarios, and has mean $1.95$ in the mixed-response scenario, corresponding to a mean absolute rank error of $0.05$.
The method therefore recovers the intended latent dimension reliably in this suite while maintaining high clustering accuracy.

\begin{table}[t]
\centering
\caption{BMLC-VI-PSM summary by simulation scenario. Accuracy, ARI, HCE, SCE, and macro F1 are averaged over 100 replications. Predictive metrics are reported only for response families present in the scenario.}
\label{tab:sim_bmlc_summary_0506}
\scriptsize
\resizebox{\linewidth}{!}{%
\begin{tabular}{lrrrrrrrrrr}
\toprule
Scenario & Acc. & Acc. SD & ARI & HCE & SCE & $\hat K$ & $\hat r$ & Rank err. & G-MSE & Brier / NB-MSE \\
\midrule
All Gaussian & 0.996 & 0.004 & 0.983 & 0.004 & 0.004 & 2.00 & 2.00 & 0.00 & 0.575 & -- \\
All binary & 0.917 & 0.047 & 0.703 & 0.083 & 0.083 & 2.00 & 2.00 & 0.00 & -- & 0.080 / -- \\
All count & 0.974 & 0.025 & 0.899 & 0.026 & 0.026 & 2.00 & 2.00 & 0.00 & -- & -- / 21403.64 \\
One Gaussian, one binary, one count & 0.954 & 0.054 & 0.836 & 0.046 & 0.046 & 2.00 & 1.95 & 0.05 & 1.453 & 0.088 / 10042.11 \\
\bottomrule
\end{tabular}
}%
\end{table}

Table~\ref{tab:sim_gain_0506} compares BMLC-VI-PSM with the best non-BMLC baseline in each scenario, where the non-BMLC set includes $K$-means, fixed-$K$ \texttt{mclust}, PCA-$K$-means, and \texttt{rrMix} in the all-Gaussian case.
Positive values indicate improvement over the strongest available non-BMLC comparator.
The method improves over the best generic baseline in the all-count and mixed-response scenarios, is essentially tied with \texttt{rrMix} in the all-Gaussian scenario, and trails response-only $K$-means in the all-binary scenario.

\begin{table}[t]
\centering
\caption{Gain of BMLC-VI-PSM over the best non-BMLC baseline in each simulation scenario. Positive values indicate improvement of the proposed method over the strongest available baseline.}
\label{tab:sim_gain_0506}
\resizebox{0.92\linewidth}{!}{%
\begin{tabular}{lccc}
\toprule
Scenario & Accuracy gain & ARI gain & Macro-F1 gain \\
\midrule
One Gaussian, one binary, one count & 0.010 & 0.041 & 0.010 \\
All count & 0.005 & 0.017 & 0.005 \\
All Gaussian & -0.001 & -0.003 & -0.001 \\
All binary & -0.014 & -0.042 & -0.014 \\
\bottomrule
\end{tabular}
}%
\end{table}

For the all-Gaussian benchmark, BMLC-VI-PSM and \texttt{rrMix} have nearly identical mean accuracy and macro F1.
The \texttt{rrMix} benchmark has a modestly higher mean ARI and lower Gaussian predictive MSE in this purely Gaussian setting, as expected for a method specialized to Gaussian reduced-rank mixtures.

\begin{table}[t]
\centering
\caption{All-Gaussian comparison with the reduced-rank mixture benchmark.}
\label{tab:sim_rrmix_0506}
\begin{tabular}{lrrrrr}
\toprule
Method & Accuracy & ARI & Macro F1 & G-MSE & $\hat r$ \\
\midrule
BMLC-VI-PSM & 0.996 & 0.983 & 0.996 & 0.575 & 2 \\
rrMix & 0.996 & 0.986 & 0.996 & 0.449 & 3 \\
\bottomrule
\end{tabular}
\end{table}

Figures~\ref{fig:sim_embedding_fitted_0506} and \ref{fig:sim_embedding_true_0506} show an example two-dimensional embedding for the all-Gaussian benchmark.
The fitted and true cluster colorings agree nearly exactly, illustrating that the PSM-based representation recovers the latent partition when the two regression components are well separated.

\begin{figure}[t]
\centering
\includeReadableFigure[width=0.78\linewidth]{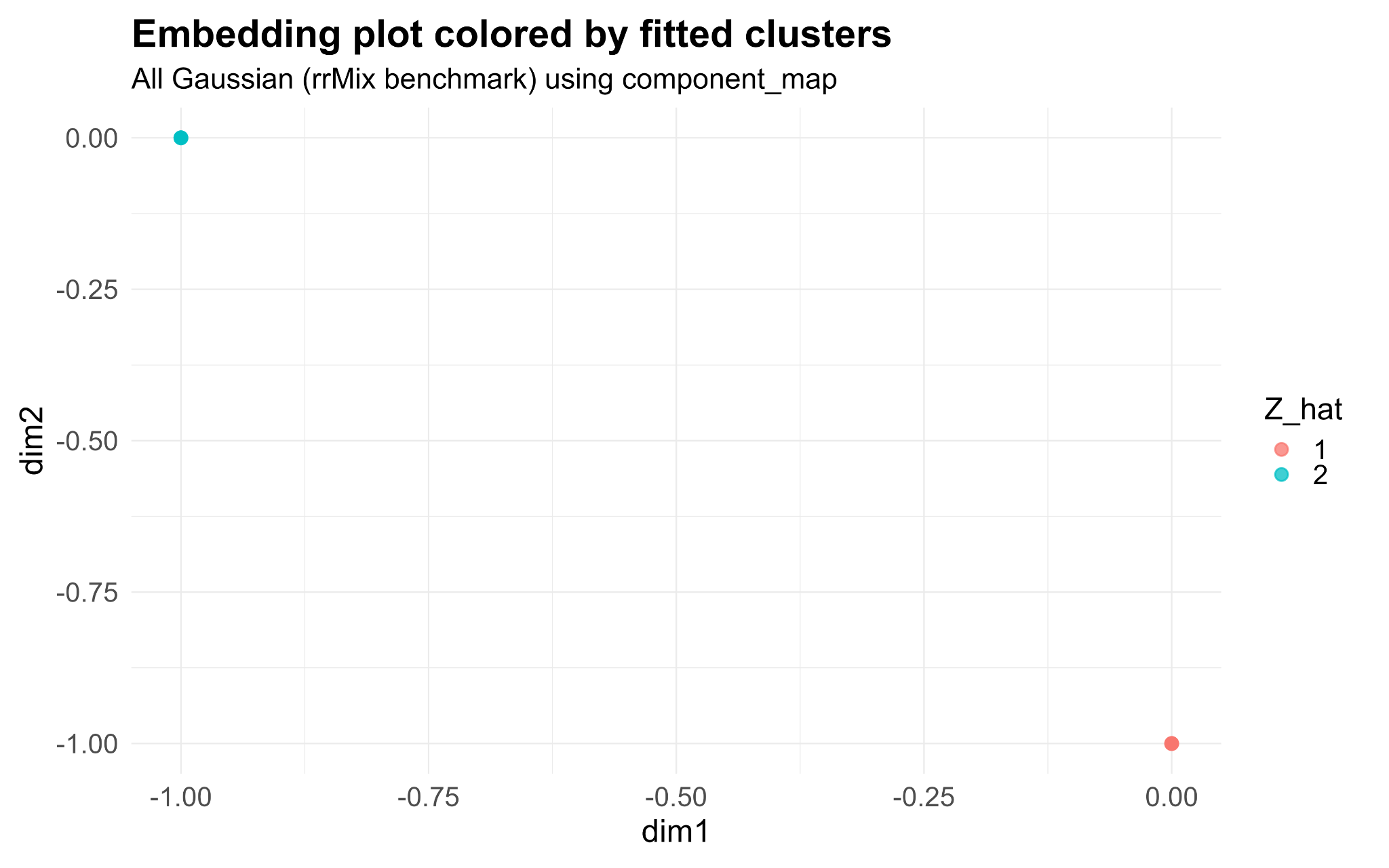}
\caption{Example embedding diagnostic from the all-Gaussian benchmark, colored by fitted clusters.}
\label{fig:sim_embedding_fitted_0506}
\end{figure}

\begin{figure}[t]
\centering
\includeReadableFigure[width=0.78\linewidth]{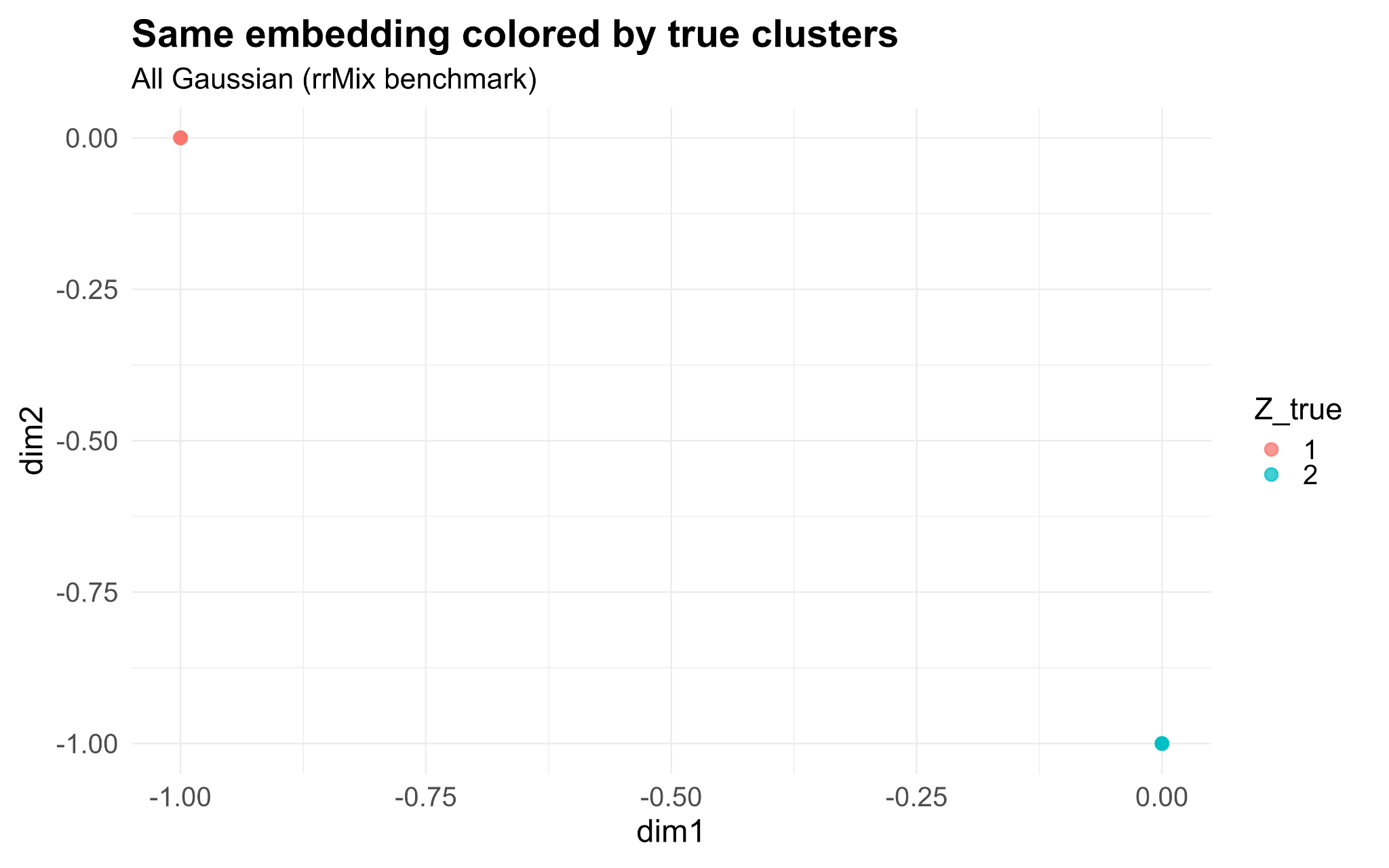}
\caption{The same all-Gaussian benchmark embedding colored by the true clusters.}
\label{fig:sim_embedding_true_0506}
\end{figure}

\section{Real-data Illustration: DoctorVisits with Gaussian, Bernoulli, and Negative Binomial Responses}\label{sec:doctorvisits}

We analyze the \texttt{DoctorVisits} data from the \texttt{AER} package as an individual-level mixed-response illustration.
The analysis uses the same response construction as in the preceding implementation and fits the model under the variational model-comparison pipeline.
The analysis sample contains $n=5190$ adults and $p=5$ standardized predictors.

\subsection{Data description and response specification}\label{sec:dv-data-0506}

For individual $i$, the response vector is
\[
\bmY_i=(Y^{(G)}_i,Y^{(B)}_i,Y^{(NB)}_i)^\top.
\]
The Gaussian coordinate is a robustly centered and scaled health score, the Bernoulli coordinate is an indicator of private health insurance, and the negative binomial coordinate is the count of doctor visits.
The predictor vector includes \texttt{age}, \texttt{income}, \texttt{illness}, \texttt{reduced}, and \texttt{nchronic}, centered and scaled when continuous.
The cluster-specific intercepts $\bmu_k$ absorb the response-level baseline terms, so no separate unpenalized constant is included in $\bmX_i$ for this run.

\subsection{Model fitting and WAIC-based selection}\label{sec:dv-selection-0506}

We fit candidate variational models over $K\in\{1,2,3\}$ and $\rmax\in\{1,2\}$, using the parsimonious WAIC selection rule and a minimum-cluster safeguard.
Table~\ref{tab:dv_waic_grid_0506} reports the WAIC grid.
The selected specification is $\hat K=2$ and $\hat\rmax=2$, with $\WAIC_{\mathrm{VI}}=31756.88$.
The next best model is the single-cluster rank-two model, with $\Delta\WAIC=1026.63$, indicating a substantial improvement from allowing latent heterogeneity.

\begin{table}[t]
\centering
\caption{DoctorVisits: $\WAIC_{\mathrm{VI}}$ grid. The selected model is chosen by the parsimonious WAIC rule with a minimum-mixture safeguard.}
\label{tab:dv_waic_grid_0506}
\resizebox{\linewidth}{!}{%
\begin{tabular}{rrrrrrr}
\toprule
$K$ & $\rmax$ & WAIC & lppd & $p_{\mathrm{WAIC}}$ & MinProp & $\Delta$WAIC \\
\midrule
\textbf{2} & \textbf{2} & \textbf{31756.88} & \textbf{-15377.64} & \textbf{500.80} & \textbf{0.106} & \textbf{0.00} \\
1 & 2 & 32783.51 & -16391.76 & 0.00 & 1.000 & 1026.63 \\
3 & 2 & 33421.24 & -13735.37 & 2975.25 & 0.071 & 1664.36 \\
1 & 1 & 34092.91 & -17046.45 & 0.00 & 1.000 & 2336.03 \\
2 & 1 & 36204.26 & -17094.11 & 1008.02 & 0.170 & 4447.38 \\
3 & 1 & 39976.78 & -14391.61 & 5596.78 & 0.069 & 8219.90 \\
\bottomrule
\end{tabular}
}%
\end{table}

Figures~\ref{fig:dv_delta_waic_0506} and \ref{fig:dv_gaussian_mse_0506} visualize the WAIC differences and Gaussian predictive error across representative DoctorVisits specifications.
The selected latent-cluster model is the WAIC winner and has substantially lower Gaussian predictive error than the single-cluster and rank-one alternatives.

\begin{figure}[t]
\centering
\includeReadableFigure[width=0.92\linewidth]{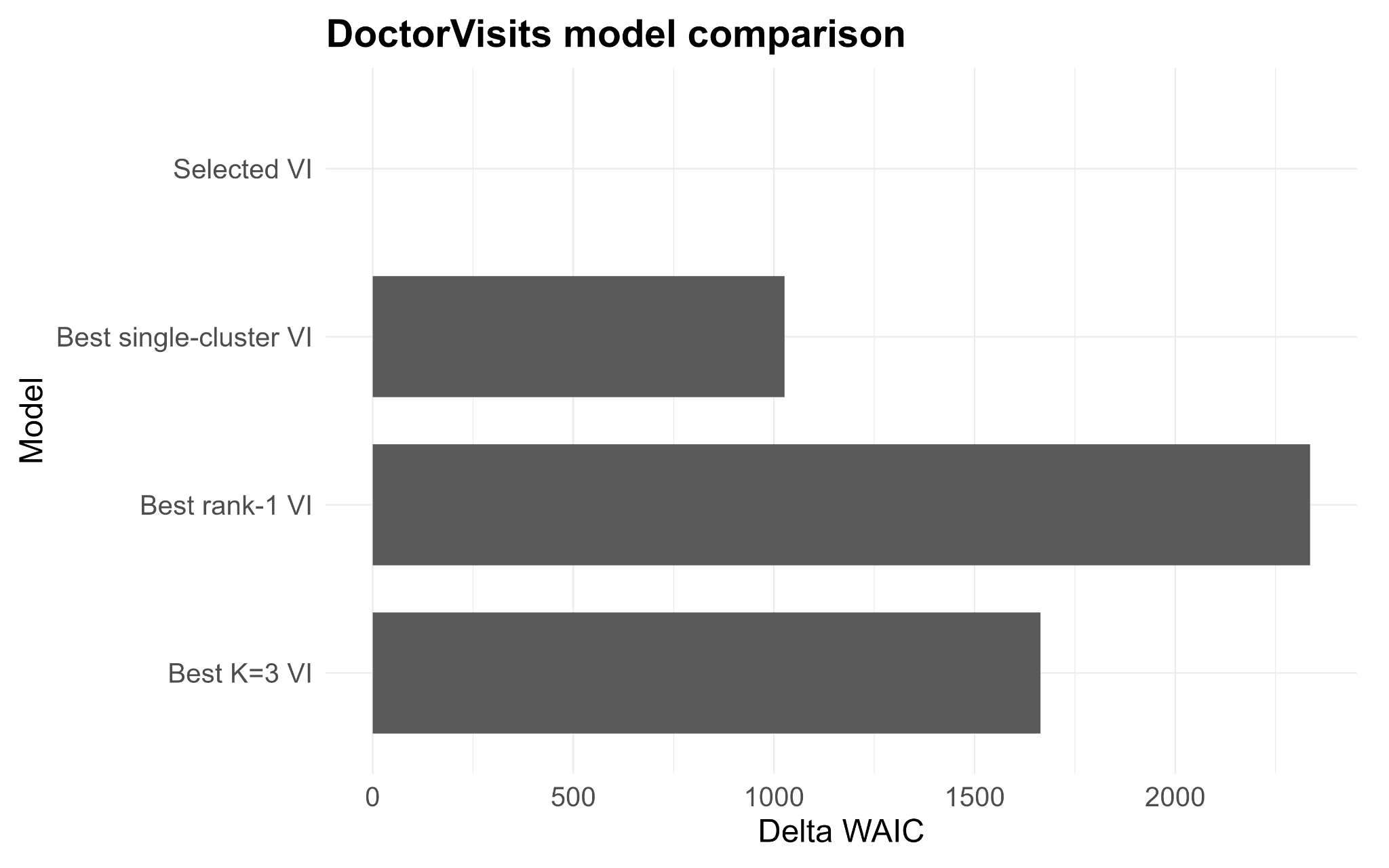}
\caption{DoctorVisits model comparison: $\Delta$WAIC relative to the selected VI model.}
\label{fig:dv_delta_waic_0506}
\end{figure}

\begin{figure}[t]
\centering
\includeReadableFigure[width=0.92\linewidth]{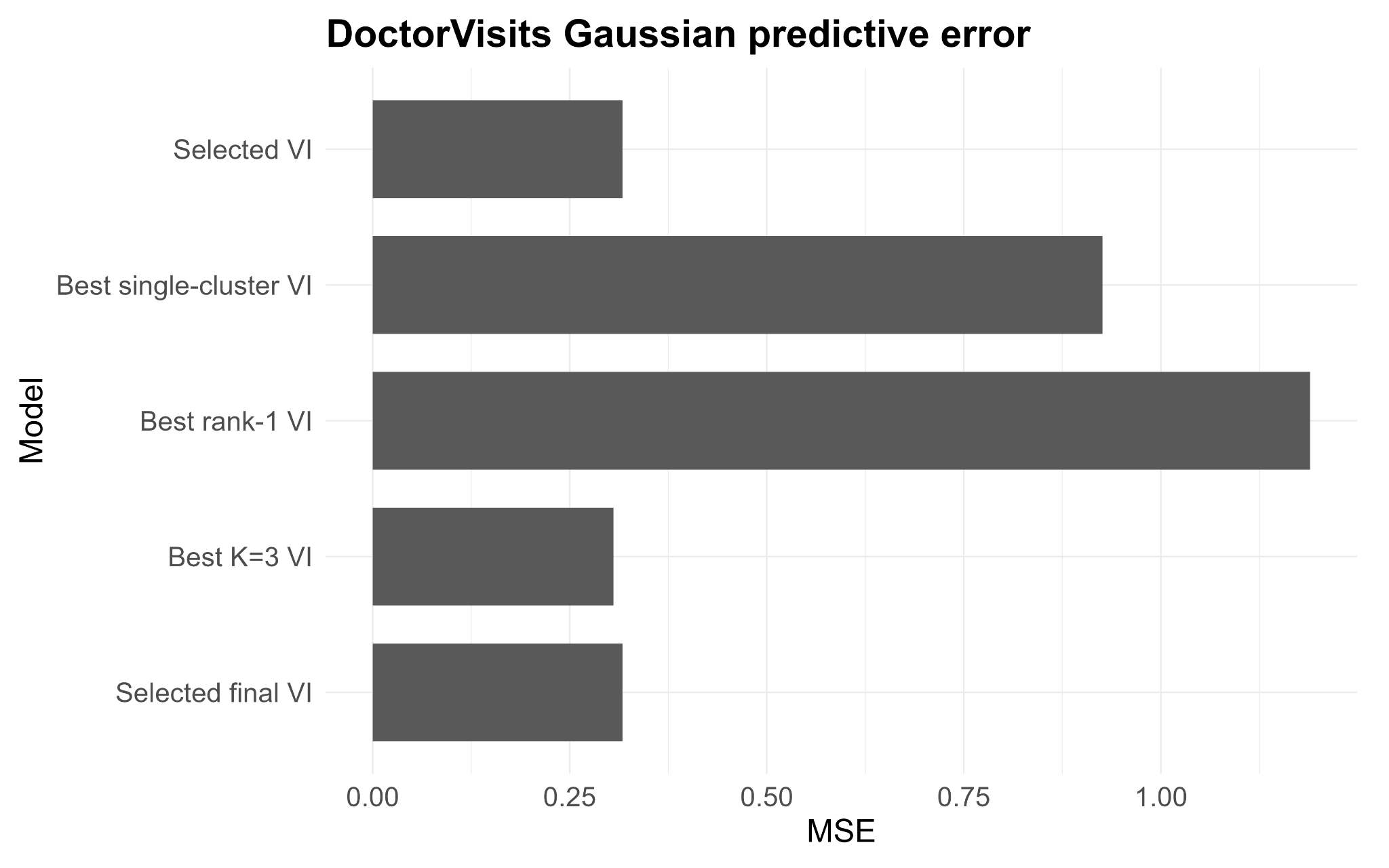}
\caption{DoctorVisits Gaussian predictive mean squared error for representative specifications.}
\label{fig:dv_gaussian_mse_0506}
\end{figure}

\subsection{Latent clusters and embedding}\label{sec:dv-clusters-0506}

The selected model identifies two latent groups.
Cluster~1 contains approximately $89.5\%$ of the sample and has lower mean health score and lower mean visit count.
Cluster~2 contains approximately $10.5\%$ of the sample and has worse health, lower private-insurance prevalence, and higher mean utilization.
Mean posterior membership certainty remains high in both clusters, with cluster-specific mean maximum probabilities $0.989$ and $0.919$.

\begin{table}[t]
\centering
\caption{DoctorVisits: cluster summary under the selected model ($\hat K=2$, $\hat\rmax=2$). Health is on the robustly scaled Gaussian scale, Private is the empirical private-insurance proportion, Visits is the empirical mean doctor-visit count, and MaxProb is the mean maximum posterior membership probability.}
\label{tab:dv_cluster_summary_0506}
\begin{tabular}{rrrrrrr}
\toprule
Cluster & $n$ & Weight & Mean Health & Mean Private & Mean Visits & Mean MaxProb \\
\midrule
1 & 4644 & 0.894 & 0.315 & 0.465 & 0.265 & 0.989 \\
2 &  546 & 0.106 & 3.107 & 0.255 & 0.612 & 0.919 \\
\bottomrule
\end{tabular}
\end{table}

\begin{figure}[t]
\centering
\includeReadableFigure[width=0.92\linewidth]{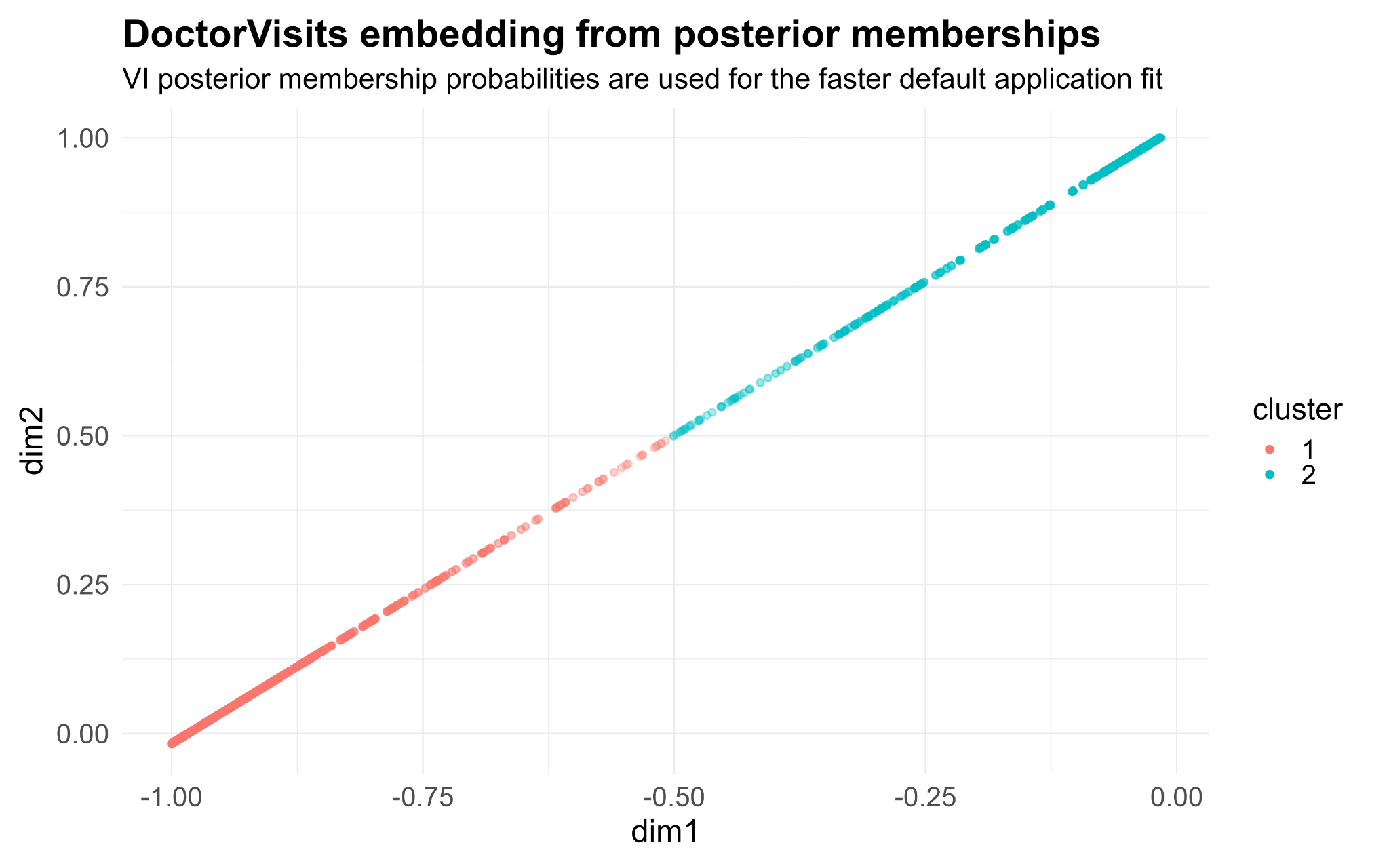}
\caption{DoctorVisits embedding from posterior membership probabilities under the selected variational model. Points are colored by maximum-posterior cluster assignment.}
\label{fig:dv_embedding_0506}
\end{figure}

\subsection{Predictive summaries}\label{sec:dv-predictive-0506}

Table~\ref{tab:dv_predictive_0506} reports predictive summaries for the selected model.
The Gaussian health response has MSE $0.317$ on the robustly scaled scale.
For the private-insurance response, the Brier score is $0.267$ and the hard misclassification rate is $0.357$ under the fitted probability thresholding rule.
For the doctor-visit count, the negative binomial predictive MSE is $0.687$.

\begin{table}[t]
\centering
\caption{DoctorVisits: predictive summaries for the selected model.}
\label{tab:dv_predictive_0506}
\begin{tabular}{rrrrr}
\toprule
Gaussian MSE & Bernoulli Brier & Bernoulli accuracy & Bernoulli misclassification & NB MSE \\
\midrule
0.317 & 0.267 & 0.643 & 0.357 & 0.687 \\
\bottomrule
\end{tabular}
\end{table}

The DoctorVisits analysis is used primarily for model comparison and clustering.
Accordingly, coefficient intervals are not reported from this variational-only run.
For final inferential reporting of cluster-specific regression effects, the Gibbs sampler in Section~\ref{sec:gibbs} can be initialized at the selected variational fit and used to produce posterior intervals on the identifiable coefficient matrices $\bB_k$.

\section{Application to Florida County COVID-19 Surveillance}\label{sec:florida}

We next analyze county-level COVID-19 surveillance outcomes in Florida using the bivariate Gaussian plus negative-binomial specialization of the proposed model.
This example differs from the individual-level \texttt{DoctorVisits} analysis in sample size, scale, and interpretation.
There are only 67 observational units, county exposures vary substantially, and the available county-level covariates describe overlapping features of testing intensity, age composition, race and ethnicity, vaccination, policy environment, and geographic density.
The analysis therefore emphasizes exposure-adjusted count modeling, low-rank regularization under strong collinearity, and label-invariant spatial reporting through county maps.

\subsection{Data assembly, outcomes, and preprocessing}\label{sec:fl-data}

The analysis combines county-level COVID-19 hospitalization and mortality summaries with county policy summaries and vaccination information.
County names are standardized before merging.
When an external population table is available, it is used as the exposure denominator; otherwise the analysis uses the best available positive exposure proxy from the supplied surveillance files.
County geometries from \texttt{tigris} are used only to display fitted summaries and latent clusters; the model itself does not include a spatial prior or adjacency penalty.

For county $i$, the two outcomes are
\[
Y^{(G)}_i=100\,\frac{\text{hospitalizations}_i}{E_i},
\;
Y^{(NB)}_i=\text{death count}_i,
\]
where $E_i$ is the county exposure used in the count offset.
The hospitalization burden is robustly centered and scaled and is modeled with a Gaussian likelihood.
The death count is modeled with a negative-binomial likelihood and offset $\log E_i$.
The active response sets are therefore $\JG=\{1\}$ and $\JNB=\{2\}$.
Predictors include log cases, log tests, test positivity, log exposure, age-specific case shares, race and ethnicity case shares, a vaccination-dose proxy, a density proxy, and policy-intensity indicators.
All predictors are median-imputed and standardized.
Table~\ref{tab:fl_data_model_actual} summarizes the specification.

\begin{table}[t]
\centering
\caption{Florida COVID-19 surveillance analysis specification. The offset in the negative-binomial arm separates county exposure from residual mortality burden.}
\label{tab:fl_data_model_actual}
\small
\resizebox{\linewidth}{!}{%
\begin{tabular}{ll}
\toprule
Item & Specification \\
\midrule
Observational units & 67 Florida counties \\
Gaussian response & Hospitalization burden per exposure unit, robustly centered and scaled \\
Negative-binomial response & County death count \\
Count offset & $\log(E_i)$, using population or the recorded exposure proxy \\
Predictors & log cases, log tests, positivity, log exposure, age/race/ethnicity shares \\
Additional predictors & vaccination proxy, density proxy, strict/medium/total policy indicators \\
Preprocessing & county-name harmonization, median imputation, predictor standardization \\
Map geometry & Florida county boundaries; used only for visualization \\
\bottomrule
\end{tabular}
}%
\end{table}

\subsection{Model fitted to Florida counties}\label{sec:fl-model}

For county $i$ in latent cluster $k$, the fitted bivariate model uses
\begin{equation}\label{eq:florida_linear_predictors}
\eta_{i1k}=\mu_{k1}+\bmX_i^\top\bB_{k,\cdot 1},
\;
\eta_{i2k}=\mu_{k2}+\bmX_i^\top\bB_{k,\cdot 2}+\log E_i.
\end{equation}
The first linear predictor enters the Gaussian hospitalization model after robust scaling.
The second enters the negative-binomial death-count model.
The offset prevents raw county size or reporting volume from being absorbed into an unrestricted coefficient.
Each component-specific coefficient matrix is factorized as $\bB_k=\bL_k\bR_k^\top$.

Candidate models are fit over
\[
K\in\{1,2,3\},\;
\rmax\in\{1,2,3\},\;
r_{NB}\in\{5,10\},
\]
with additional grid values for the Dirichlet concentration and ridge stabilization constants.
The selected working specification is a three-cluster, rank-two model,
\[
\hat K=3,
\;
\hat\rmax=2.
\]
This choice retains enough heterogeneity to separate county-level surveillance profiles while avoiding an unstable unrestricted mixture in a 67-county sample.

\subsection{Florida posterior prediction, clusters, and maps}\label{sec:fl-results}

Figure~\ref{fig:fl_covid_maps} compares observed and posterior predictive summaries for the two response coordinates.
The cluster map is intentionally not included in this panel; it is displayed separately in Figure~\ref{fig:fl_covid_cluster_map}.
The posterior predictive hospitalization maps retain the statewide burden gradient while smoothing isolated county extremes.
The death-count maps retain the ordering of high-count counties but reduce the influence of small-county spikes through the negative-binomial variance and the exposure offset.

\begin{figure}[t]
\centering
\includeReadableFigure[width=0.98\linewidth]{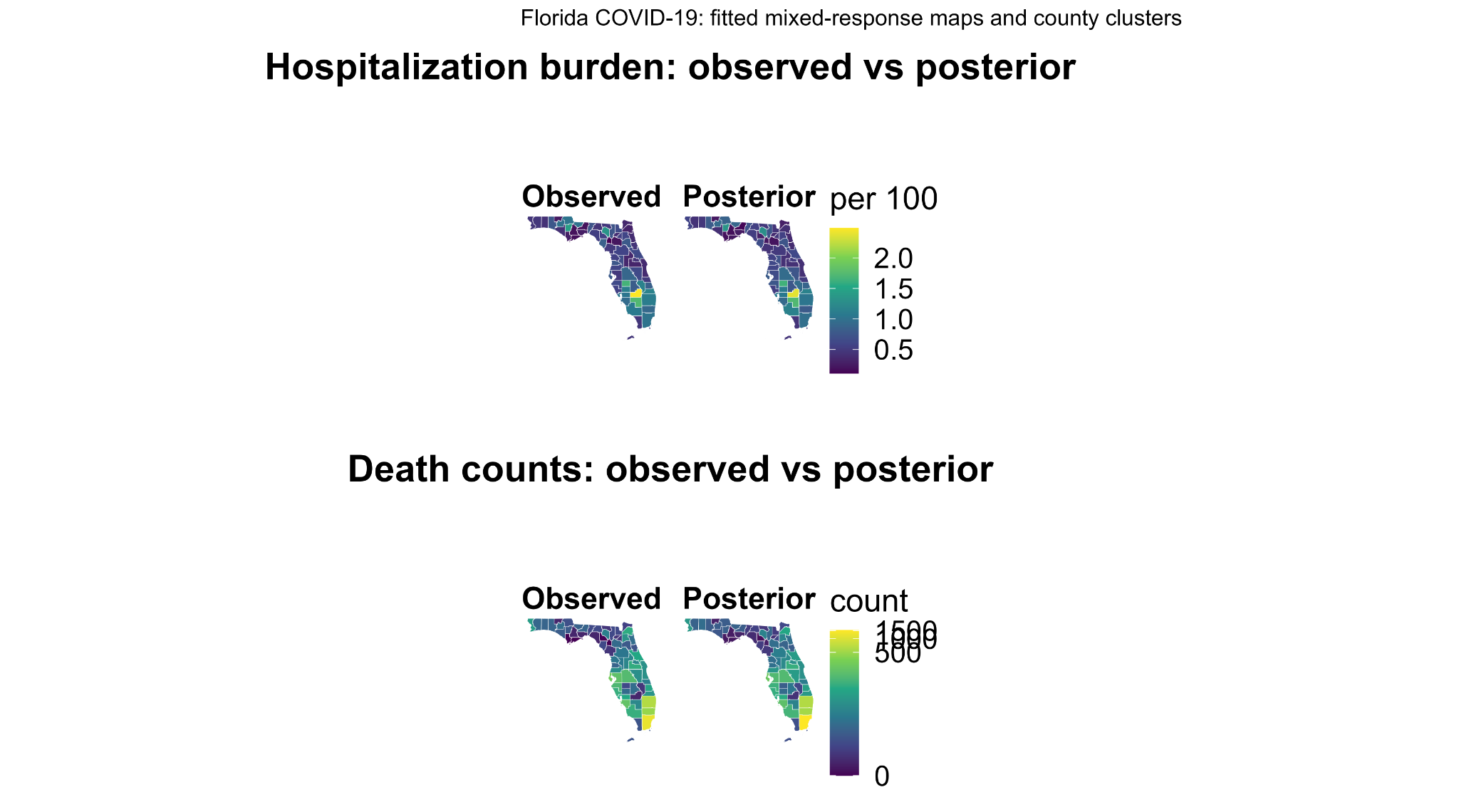}
\caption{Florida county COVID-19 response maps. The first row compares observed and posterior predictive hospitalization burden; the second row compares observed and posterior predictive death counts. The latent county partition is shown separately in Figure~\ref{fig:fl_covid_cluster_map}.}
\label{fig:fl_covid_maps}
\end{figure}

\begin{figure}[t]
\centering
\includeReadableFigure[width=0.82\linewidth]{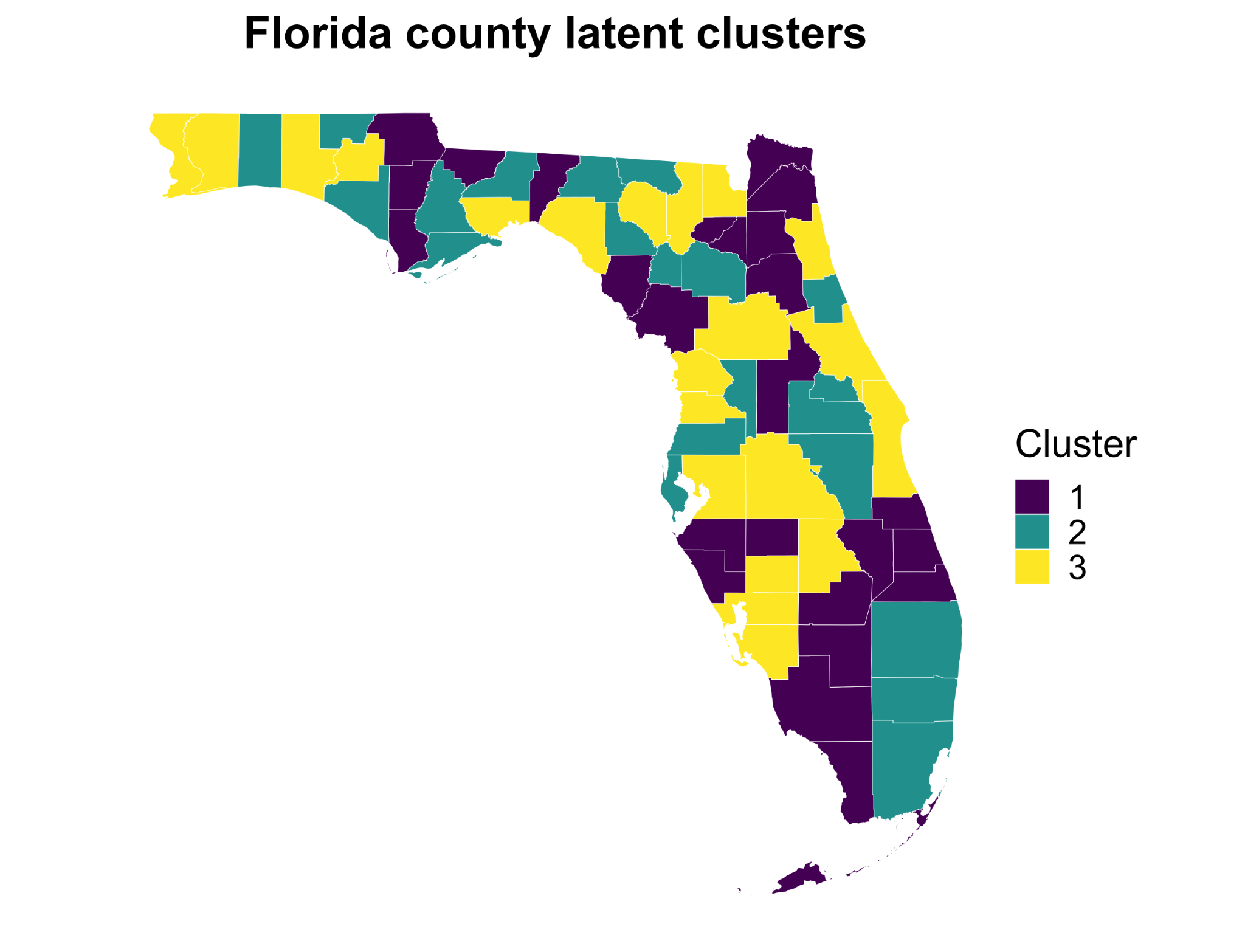}
\caption{Florida county latent-cluster map. The partition is inferred jointly from the Gaussian hospitalization arm and the negative-binomial death-count arm.}
\label{fig:fl_covid_cluster_map}
\end{figure}

Figure~\ref{fig:fl_nb_diagnostic} provides a response-specific diagnostic for the negative-binomial arm on the $\log(1+y)$ scale.
The fitted values follow the observed ordering over most of the range, while extreme counties are shrunk toward the fitted component structure.
The county partition in Figure~\ref{fig:fl_covid_cluster_map} is spatially interpretable even though spatial adjacency is not used in the likelihood.
Thus the map should be read as a diagnostic visualization of joint surveillance profiles, not as a spatially smoothed estimator.

\begin{figure}[t]
\centering
\includeReadableFigure[width=0.86\linewidth]{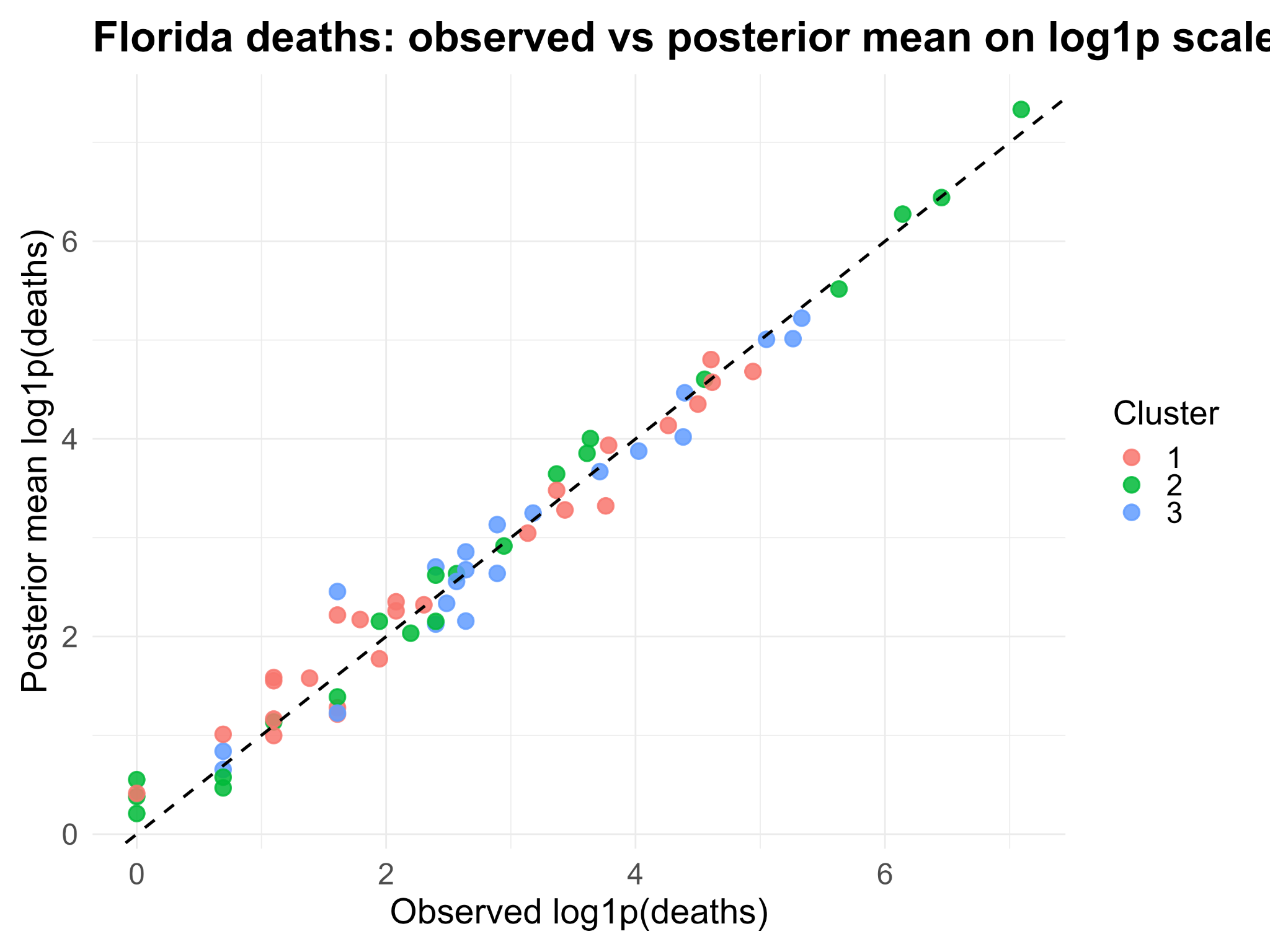}
\caption{Florida death-count diagnostic. Points compare observed death counts and posterior mean death counts on the $\log(1+y)$ scale and are colored by fitted county cluster.}
\label{fig:fl_nb_diagnostic}
\end{figure}

\subsection{Florida coefficient interpretation}\label{sec:fl-coef}

Table~\ref{tab:fl_coef} gives representative coefficient summaries from a Laplace approximation centered at the variational fit.
The coefficients are cluster-specific associations on the working scales and should not be interpreted causally.
The count arm emphasizes variables that summarize infection burden and exposure-adjusted county structure, while the Gaussian arm captures residual hospitalization burden after robust standardization.
A predictor may therefore have different magnitudes, or even different signs, across the two response arms.
This is expected in a mixed-response surveillance model: absolute mortality burden and hospitalization burden per exposure are related but not interchangeable endpoints.

\begin{table}[t]
\centering
\caption{Florida COVID-19 application. Posterior medians and approximate 95\% credible intervals for selected regression coefficients under the selected model with $K=3$ and $\rmax=2$. Intervals excluding zero are shown in bold.}
\label{tab:fl_coef}
\scriptsize
\begin{tabular}{lllrrr}
\toprule
Cluster & Response & Predictor & Median & 95\% CI lower & 95\% CI upper \\
\midrule
1 & Deaths (NB, log mean) & \texttt{LogDeaths} & \textbf{3.42} & \textbf{2.97} & \textbf{3.83} \\
1 & Hosp. rate (Gaussian) & \texttt{LogHosp} & \textbf{-3.99} & \textbf{-5.64} & \textbf{-2.41} \\
2 & Deaths (NB, log mean) & \texttt{LogDeaths} & \textbf{3.41} & \textbf{3.22} & \textbf{3.60} \\
3 & Deaths (NB, log mean) & Intercept & \textbf{20.30} & \textbf{16.30} & \textbf{23.80} \\
3 & Deaths (NB, log mean) & \texttt{Pop} & \textbf{-0.91} & \textbf{-1.13} & \textbf{-0.69} \\
3 & Deaths (NB, log mean) & \texttt{LogDeaths} & \textbf{0.72} & \textbf{0.48} & \textbf{0.96} \\
3 & Hosp. rate (Gaussian) & Intercept & \textbf{17.80} & \textbf{14.10} & \textbf{21.60} \\
3 & Hosp. rate (Gaussian) & \texttt{LogHosp} & \textbf{-22.60} & \textbf{-25.00} & \textbf{-20.10} \\
3 & Hosp. rate (Gaussian) & \texttt{Pop} & \textbf{1.57} & \textbf{1.22} & \textbf{1.94} \\
\bottomrule
\end{tabular}
\end{table}

The Florida analysis illustrates why the proposed mixed-response formulation is useful for small-area surveillance.
A Gaussian-only analysis would ignore overdispersed death counts, a count-only analysis would discard hospitalization burden, and separate regressions would not yield a single label-invariant county partition.
The selected low-rank mixture combines the two response types while stabilizing estimation in a small, highly collinear county-level design.

\section{Application to U.S. States Influenza Surveillance (CDC FluView)}\label{sec:flu}

The third application analyzes state-level influenza surveillance summaries using the same bivariate Gaussian plus negative-binomial formulation.
The scientific distinction between the two response coordinates is important.
Weighted ILI is a relative activity measure, whereas the total number of ILI patient visits is an absolute burden measure affected by state size and reporting intensity.
A state can therefore have many ILI patient visits without having the highest weighted ILI percentage, and conversely a smaller state can have high relative activity with moderate absolute counts.

\subsection{Data assembly, outcomes, and preprocessing}\label{sec:flu-data}

The analysis uses a completed-season state-level summary when one is supplied.
If weekly ILINet-style observations are supplied, they are aggregated by state to obtain the season-level mean weighted ILI percentage and the total number of ILI patient visits.
State geometries from \texttt{tigris} are used only for mapping.
The Gaussian response is $\log(1+\mathrm{weighted\ ILI})$, robustly centered and scaled.
The negative-binomial response is the season-level total number of ILI patient visits, rounded to counts in thousands.
The count arm uses a log exposure offset based on state size or reporting volume when available.
Predictors are numeric state-level surveillance, demographic, and vulnerability covariates after excluding response and exposure variables; they are median-imputed and standardized.

\begin{table}[t]
\centering
\caption{U.S. influenza surveillance analysis specification. The two endpoints represent relative activity and absolute patient burden.}
\label{tab:flu_data_model_actual}
\small
\resizebox{\linewidth}{!}{%
\begin{tabular}{ll}
\toprule
Item & Specification \\
\midrule
Observational units & U.S. states in the seasonal surveillance summary \\
Gaussian response & $\log(1+\mathrm{weighted\ ILI})$, robustly centered and scaled \\
Negative-binomial response & Season-level ILI patient visits, rounded to thousands \\
Count offset & log state-size or reporting-size exposure \\
Predictors & numeric state-level surveillance, demographic, and vulnerability covariates \\
Preprocessing & seasonal aggregation if needed, median imputation, predictor standardization \\
Map geometry & U.S. state boundaries; used only for visualization \\
\bottomrule
\end{tabular}
}%
\end{table}

\subsection{Model fitted to U.S. influenza summaries}\label{sec:flu-model}

For state $i$ and latent cluster $k$, the fitted linear predictors are
\begin{equation}\label{eq:flu_linear_predictors}
\eta_{i1k}=\mu_{k1}+\bmX_i^\top\bB_{k,\cdot 1},
\;
\eta_{i2k}=\mu_{k2}+\bmX_i^\top\bB_{k,\cdot 2}+\log E_i,
\end{equation}
where $E_i$ is the exposure used in the count arm.
Candidate models are fit over
\[
K\in\{1,2,3\},\;
\rmax\in\{1,2,3\},\;
r_{NB}\in\{5,10\},
\]
with the same ridge and Dirichlet-prior grid used in the Florida application.
The reported bivariate influenza fit selects a low-rank model with a small number of state clusters; the rank-two specification is retained for interpretation because it separates broad burden intensity from differences between relative activity and absolute patient volume.
Table~\ref{tab:waic_us_bivar} gives representative rank comparisons for the two-component bivariate fit.

\begin{table}[t]
\centering
\caption{U.S. influenza application. Representative $\mathrm{WAIC}_{\mathrm{VI}}$ values over factor ranks for the bivariate Gaussian plus negative-binomial model with $K=2$ mixture components.}
\label{tab:waic_us_bivar}
\begin{tabular}{lccc}
\toprule
$K$ & $\rmax$ & $\mathrm{WAIC}_{\mathrm{VI}}$ & SE($\mathrm{WAIC}_{\mathrm{VI}}$) \\
\midrule
2 & 1 & 231.0 & 16.3 \\
2 & 2 & 165.0 & 8.82 \\
2 & 3 & 166.0 & 8.62 \\
\bottomrule
\end{tabular}
\end{table}

\subsection{U.S. influenza posterior prediction, clusters, and maps}\label{sec:flu-results}

Figure~\ref{fig:us_flu_bivar_maps} compares observed and posterior predictive maps for the two influenza endpoints.
The state-cluster map is deliberately excluded from this response panel and is shown separately in Figure~\ref{fig:us_flu_cluster_map}.
The posterior maps preserve broad regional patterns in both weighted ILI and patient-visit burden while smoothing the largest state totals.
This smoothing is appropriate for an exposure-adjusted negative-binomial mixture because large absolute counts can arise from population size, reporting intensity, or genuine influenza activity.

\begin{figure}[t]
\centering
\includeReadableFigure[width=0.98\linewidth]{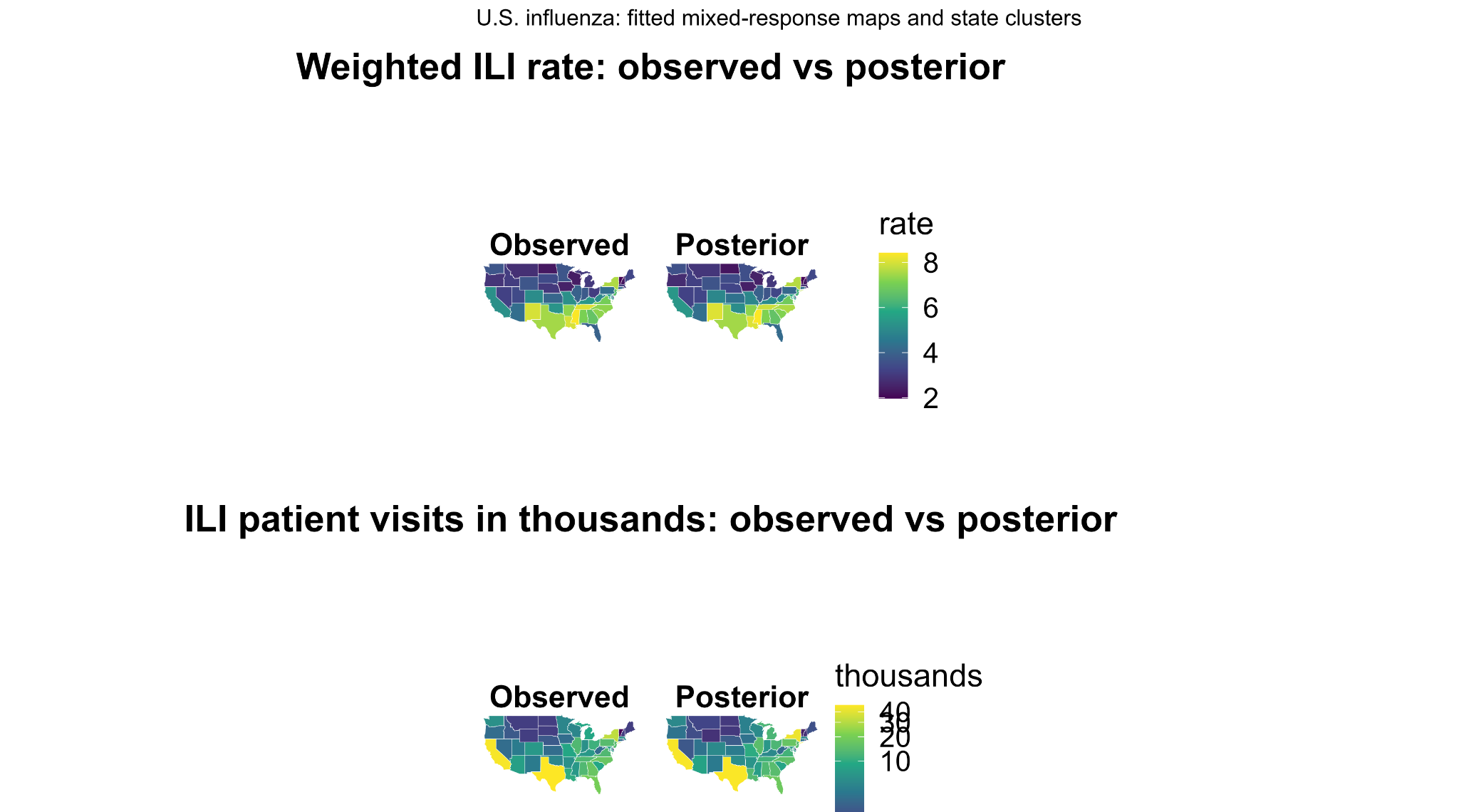}
\caption{U.S. influenza response maps. The first row compares observed and posterior predictive weighted ILI rates; the second row compares observed and posterior predictive ILI patient counts in thousands. The latent state partition is shown separately in Figure~\ref{fig:us_flu_cluster_map}.}
\label{fig:us_flu_bivar_maps}
\end{figure}

\begin{figure}[t]
\centering
\includeReadableFigure[width=0.82\linewidth]{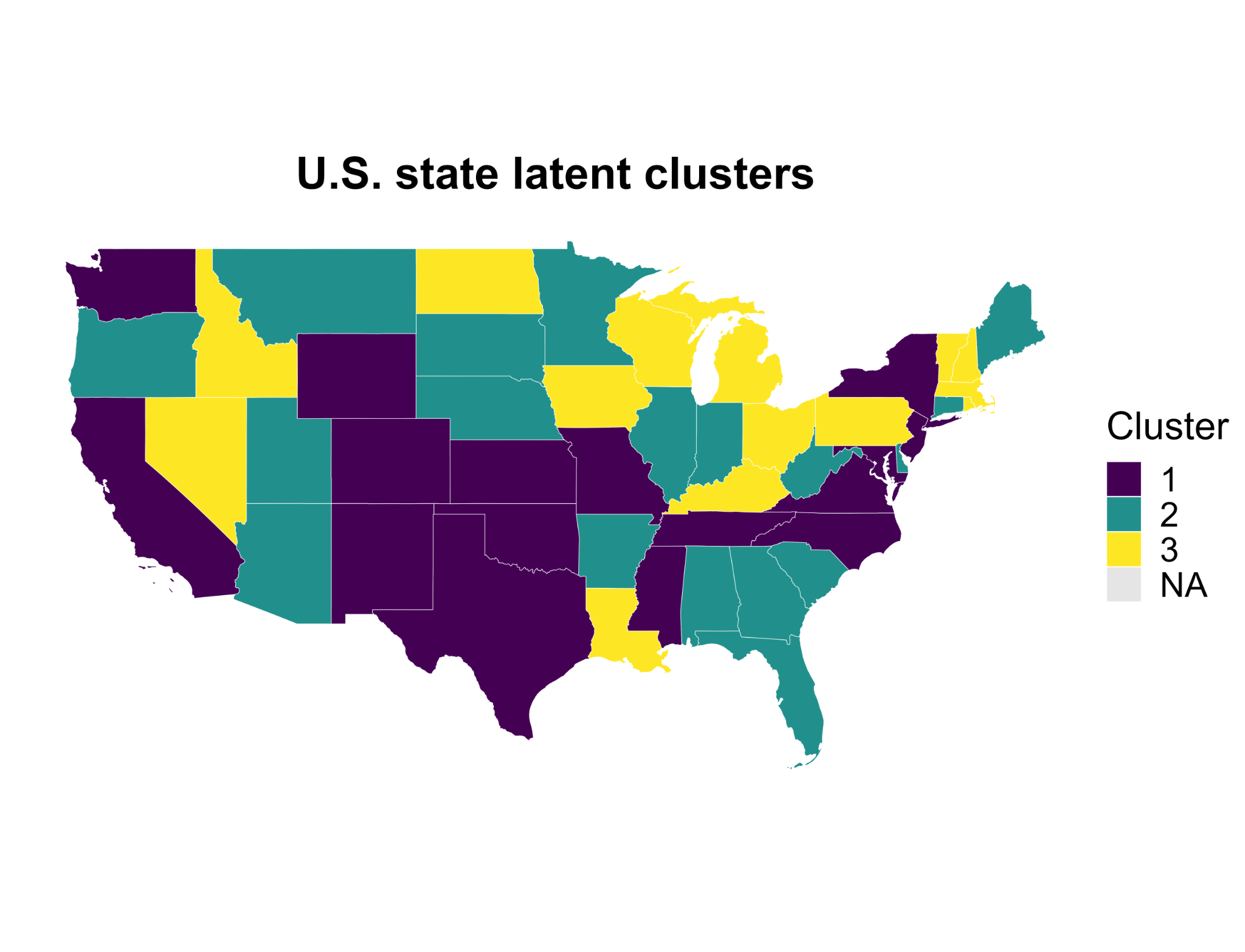}
\caption{U.S. state latent-cluster map. The partition is inferred jointly from transformed weighted ILI activity and negative-binomial ILI patient-count burden.}
\label{fig:us_flu_cluster_map}
\end{figure}

Figure~\ref{fig:us_nb_diagnostic} gives a log-scale diagnostic for the patient-count arm.
The fitted posterior means track the observed ordering over most states, with shrinkage at the extremes due to the negative-binomial variance and component-level sharing.
The clusters in Figure~\ref{fig:us_flu_cluster_map} therefore summarize joint rate-count behavior rather than a single high-rate or high-count endpoint.

\begin{figure}[t]
\centering
\includeReadableFigure[width=0.86\linewidth]{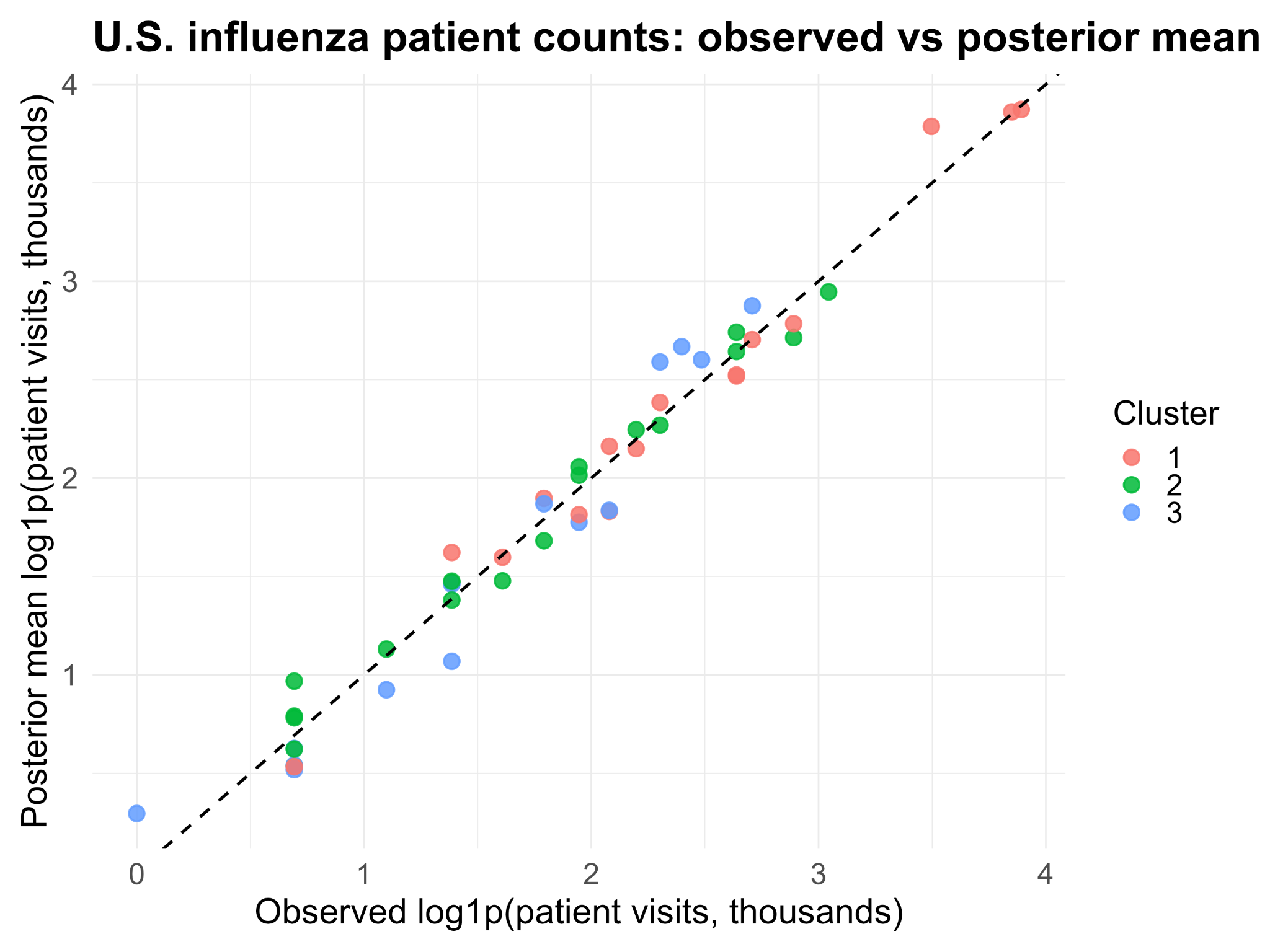}
\caption{U.S. influenza patient-count diagnostic. Points compare observed and posterior mean ILI patient visits in thousands on the $\log(1+y)$ scale and are colored by fitted state cluster.}
\label{fig:us_nb_diagnostic}
\end{figure}

\subsection{U.S. influenza coefficient interpretation}\label{sec:flu-coef}

Table~\ref{tab:us_coef} reports representative coefficient summaries.
Across both clusters, total ILI patient volume is positively associated with the negative-binomial log mean of patient counts.
The same variable is negatively associated with the transformed weighted ILI rate after the two endpoints are modeled jointly.
This is consistent with the surveillance interpretation: large absolute patient volume may reflect state size or reporting participation, whereas weighted ILI is a relative activity measure.
The low-rank mixed-response model separates these two sources of information rather than collapsing them into a single scalar risk score.

\begin{table}[t]
\centering
\caption{U.S. influenza model. Posterior medians and approximate 95\% credible intervals for selected coefficients under the selected bivariate model with $K=2$ and $\rmax=2$. Intervals excluding zero are shown in bold.}
\label{tab:us_coef}
\scriptsize
\begin{tabular}{lllrrr}
\toprule
Cluster & Response & Predictor & Median & 95\% CI lower & 95\% CI upper \\
\midrule
1 & log1p(ILI) (scaled) & Intercept & \textbf{0.104} & \textbf{0.0252} & \textbf{0.193} \\
1 & log1p(ILI) (scaled) & mean wILI & -0.0119 & -1.09 & 1.02 \\
1 & log1p(ILI) (scaled) & total ILI pts & \textbf{-0.636} & \textbf{-0.877} & \textbf{-0.415} \\
1 & Patients k (log mu) & Intercept & 0.0638 & -0.0794 & 0.199 \\
1 & Patients k (log mu) & mean wILI & 0.0912 & -2.12 & 2.42 \\
1 & Patients k (log mu) & total ILI pts & \textbf{1.46} & \textbf{1.12} & \textbf{1.91} \\
2 & log1p(ILI) (scaled) & Intercept & 0.116 & -0.00140 & 0.247 \\
2 & log1p(ILI) (scaled) & mean wILI & -0.0139 & -1.08 & 1.13 \\
2 & log1p(ILI) (scaled) & total ILI pts & \textbf{-0.414} & \textbf{-0.677} & \textbf{-0.180} \\
2 & Patients k (log mu) & Intercept & 0.0126 & -0.0904 & 0.105 \\
2 & Patients k (log mu) & mean wILI & 0.0132 & -1.17 & 1.21 \\
2 & Patients k (log mu) & total ILI pts & \textbf{0.490} & \textbf{0.345} & \textbf{0.660} \\
\bottomrule
\end{tabular}
\end{table}

The influenza application reinforces the main methodological point of the paper.
A Gaussian model alone would ignore overdispersed patient counts and exposure differences, whereas a count-only analysis would lose the rate information used in surveillance comparisons.
The selected bivariate low-rank mixture captures both views of influenza burden and reports a single label-invariant state partition.

\section{Discussion and Conclusion}\label{sec:discussion}

We proposed a Bayesian low-rank latent-cluster regression framework for multivariate mixed outcomes that combines cluster-specific low-rank coefficient matrices, P\'olya--Gamma augmentation for Bernoulli and negative binomial likelihoods, and label-invariant clustering through a posterior similarity matrix.
The model is designed for settings in which the dominant structure is not only sparsity in a high-dimensional coefficient vector, but also shared low-dimensional outcome structure and heterogeneity across observational units.

The simulation study considered four response-family regimes with 100 replications per regime: all Gaussian, all Bernoulli, all negative binomial, and mixed Gaussian--Bernoulli--negative binomial.
The proposed BMLC-VI-PSM procedure recovered the true number of clusters in all scenarios and recovered the intended rank in all but a small fraction of mixed-response replications.
Its clustering performance was essentially tied with \texttt{rrMix} in the all-Gaussian benchmark, strongest among the reported methods in the all-count and mixed-response settings, and slightly below response-only $K$-means in the all-binary setting.
These results support the intended use of the model for genuinely mixed multivariate outcomes while also identifying a setting where a simple response-space baseline can be competitive.

The three real-data examples cover complementary use cases.
In the \texttt{DoctorVisits} analysis, the selected model used two latent clusters and rank two, substantially improved WAIC relative to single-cluster and rank-one alternatives, and separated a large lower-utilization group from a smaller group with worse health, lower private-insurance prevalence, and higher visit counts.
In the Florida COVID-19 analysis, the model combines a Gaussian hospitalization-burden endpoint with an exposure-offset negative-binomial death-count endpoint and yields spatially interpretable county profiles without imposing a spatial prior.
In the U.S. influenza analysis, the model fits transformed weighted ILI rates jointly with exposure-adjusted patient counts and separates relative activity from absolute patient burden.
The variational algorithm provides a fast model-selection and initialization procedure; full posterior coefficient intervals can be obtained by running the Gibbs sampler initialized at the selected variational fit.

Several limitations remain.
First, variational inference is used as the primary screening and application engine, and its uncertainty estimates may be too narrow.
Second, the current likelihood assumes conditional independence of response coordinates given the cluster label and linear predictors; residual dependence extensions may improve fit in applications with strong unexplained cross-response association.
Third, the mixture weights are independent across units; spatial, temporal, or network-structured priors on component allocation would be natural extensions for surveillance applications.

Overall, the proposed latent-cluster low-rank framework provides a flexible approach to mixed-outcome regression, combining dimension reduction, latent heterogeneity, scalable computation, and stable label-invariant clustering.

\section*{Acknowledgments}
Hsin-Hsiung Huang's work was partially supported by the National Science Foundation under grants DMS-1924792 and DMS-2318925.

\appendix

\section{Proofs for Section~\ref{sec:theory}}\label{app:proofs}

We provide detailed proofs for Theorems~\ref{thm:post_contr_H} through \ref{thm:psm_meanshift_consistency}.
Throughout, $C,c,c'$ denote finite positive constants whose values may change from line to line.

\subsection{Step 0: The three likelihoods share quadratic curvature in the linear predictor on compacts}

For each response coordinate $j$, write the log-likelihood as $\ell_j(y;\eta,\psi_j)$.
Assumption (B4) implies that on any compact interval $I\subset\R$, the negative second derivative $-\partial_\eta^2 \ell_j$ is bounded between positive constants, in expectation under the true model.
This yields local quadratic control of Kullback--Leibler divergence in the linear predictor $\eta$.

\begin{lemma}[Quadratic KL control in $\eta$ on compacts]\label{lem:kl_eta}
Fix a response index $j$ and suppose $\eta,\eta_0\in I$, where $I\subset\R$ is compact.
Then there exist constants $0<C_1\le C_2<\infty$, depending on $I$ and the response family, such that
\[
C_1(\eta-\eta_0)^2
\le
\KL\bigl(f_j(\cdot;\eta_0,\psi_0)\,\|\,f_j(\cdot;\eta,\psi_0)\bigr)
\le
C_2(\eta-\eta_0)^2.
\]
\end{lemma}

\begin{proof}
Consider $g(\eta)=-\E_{\eta_0}[\ell_j(Y;\eta,\psi_0)]$.
Then $\KL(f_{\eta_0}\|f_\eta)=g(\eta)-g(\eta_0)$ and $g'(\eta_0)=0$.
By Taylor's theorem,
\[
g(\eta)-g(\eta_0)=\frac12(\eta-\eta_0)^2 g''(\eta^\star)
\]
for some $\eta^\star$ between $\eta$ and $\eta_0$.
Now $g''(\eta)=\E_{\eta_0}[-\partial_\eta^2 \ell_j(Y;\eta,\psi_0)]$.
Assumption (B4) bounds $g''(\eta)$ between $c_I$ and $C_I$ uniformly for $\eta\in I$, hence the claim with $C_1=c_I/2$ and $C_2=C_I/2$.
\end{proof}

\begin{lemma}[Component-wise KL control]\label{lem:kl_component}
Fix a mixture component $k$ and $x$.
Suppose $\eta_k(x)$ and $\eta_{0k}(x)$ lie in a compact set $I^q$.
Then for some finite $C$,
\[
\KL\bigl(f_{0k}(\cdot\mid x)\,\|\, f_k(\cdot\mid x)\bigr)
\le
C\|\eta_k(x)-\eta_{0k}(x)\|_2^2.
\]
\end{lemma}

\begin{proof}
Conditional on $Z=k$ the component density factorizes across $j$ by the model:
$f_k(\bm y\mid x)=\prod_{j=1}^q f_j(y_j;\eta_{kj}(x),\psi_j)$.
Thus Kullback--Leibler divergence adds across coordinates:
\[
\KL(f_{0k}\|f_k)=\sum_{j=1}^q \KL(f_{0k,j}\|f_{k,j}).
\]
Apply Lemma~\ref{lem:kl_eta} to each coordinate and sum.
\end{proof}

\subsection{Step 1: Lipschitz control of the linear predictor under bounded design}

\begin{lemma}[Lipschitz map $(\bmu,\bB)\mapsto \eta(x)$]\label{lem:eta_lip}
Under Assumption (B1), there exists $C_X<\infty$ such that $\|x\|_2\le C_X$ almost surely.
Hence for all $k$ and all $x$ in the support of $P_X$,
\[
\|\eta_k(x)-\eta_{0k}(x)\|_2
\le
\|\bmu_k-\bmu_{0k}\|_2 + C_X\|\bB_k-\bB_{0k}\|_F.
\]
\end{lemma}

\begin{proof}
By definition, $\eta_k(x)=\bmu_k+\bB_k^\top x$.
Thus
\[
\eta_k(x)-\eta_{0k}(x)
=
(\bmu_k-\bmu_{0k}) + (\bB_k-\bB_{0k})^\top x.
\]
Take Euclidean norms and use $\|A^\top x\|_2\le \|A\|_F\|x\|_2$.
Bound $\|x\|_2$ by $C_X$ from compact support.
\end{proof}

\subsection{Step 2: Prior support for low-rank matrices under the MGP prior}

\begin{lemma}[Local prior mass for $\bB_k$]\label{lem:prior_mass_B}
Fix $k$ and any matrix $\bB_{0k}\in\R^{p\times q}$ with $\rank(\bB_{0k})\le \rmax$.
Under the MGP prior in Section~\ref{sec:priors}, for all $\epsilon>0$,
\[
\Pi\bigl(\|\bB_k-\bB_{0k}\|_F<\epsilon\bigr)>0.
\]
\end{lemma}

\begin{proof}
Because $\rank(\bB_{0k})\le \rmax$, there exist $\bL_{0k}\in\R^{p\times \rmax}$ and $\bR_{0k}\in\R^{q\times \rmax}$ such that $\bB_{0k}=\bL_{0k}\bR_{0k}^\top$, padding with zero columns if necessary.
Conditional on $(\phi_k,\lambda_{k1:\rmax})$, the prior for $(\bL_k,\bR_k)$ is a nondegenerate Gaussian density on $\R^{p \rmax}\times \R^{q \rmax}$, hence assigns positive probability to
\[
\{\|\bL_k-\bL_{0k}\|_F<\delta,\ \|\bR_k-\bR_{0k}\|_F<\delta\}
\]
for any $\delta>0$.
The map $(\bL,\bR)\mapsto \bL\bR^\top$ is continuous, so choose $\delta$ small enough that
$\|\bL\bR^\top-\bL_{0k}\bR_{0k}^\top\|_F<\epsilon$ on that set.
Integrate over the hyperpriors to conclude positive induced mass on the $\epsilon$-ball in $\bB_k$.
\end{proof}

\subsection{Step 3: KL neighborhoods contain parameter-metric balls on the sieve}

Let $\Theta_n$ be the sieve from Assumption (B6) enforcing bounded norms of $\bmu_k$, $\bB_k$, and nuisance parameters and a lower bound on mixture weights.

\begin{lemma}[Local KL control by $d_\Theta$ on $\Theta_n$]\label{lem:kl_by_dtheta}
On $\Theta_n$, for all $\theta$ sufficiently close to $\theta_0$, up to permutation,
\[
\KL(P_{\theta_0}\|P_\theta)\le C\, d_\Theta^2(\theta,\theta_0).
\]
\end{lemma}

\begin{proof}
Fix an aligning permutation $\tau$ achieving $d_\Theta$.
By Lemma~\ref{lem:eta_lip}, on the bounded design set,
$\|\eta_k(x)-\eta_{0,\tau(k)}(x)\|_2$ is bounded by a constant multiple of
$\|\bmu_k-\bmu_{0,\tau(k)}\|_2+\|\bB_k-\bB_{0,\tau(k)}\|_F$.
By Lemma~\ref{lem:kl_component}, each component-wise Kullback--Leibler divergence at $x$ is bounded by
$C\|\eta_k(x)-\eta_{0,\tau(k)}(x)\|_2^2$ on compacts, and integrating over $P_X$ gives the same form.
Finally, for mixtures with weights bounded away from $0$ and separated components as in Assumption (B3), the mixture Kullback--Leibler divergence is controlled by the sum of aligned component divergences plus $\|\bpi-\bpi_0^\tau\|_1^2$.
All nuisance parameters are included in the definition of $d_\Theta$ and are bounded on the sieve.
Collect constants into $C$.
\end{proof}

\subsection{Proof of Theorem~\ref{thm:post_contr_H}: contraction in integrated Hellinger}

\begin{proof}
We verify the standard ingredients of posterior contraction for i.i.d.\ models: the existence of tests for $d_H$-separated sets, a sieve $\Theta_n$ with exponentially small prior outside, an entropy bound for $\Theta_n$ in $d_H$, and sufficient prior mass in a Kullback--Leibler neighborhood of $\theta_0$.

Assumption (B6) provides the sieve and entropy conditions.
It therefore remains to show the prior Kullback--Leibler mass condition and invoke the existence of tests.

By Lemma~\ref{lem:kl_by_dtheta}, for $\theta$ sufficiently close to $\theta_0$,
\[
\KL(P_{\theta_0}\|P_\theta)\le C\,d_\Theta^2(\theta,\theta_0).
\]
Hence a sufficiently small $d_\Theta$-ball is contained in a Kullback--Leibler ball.
Lemma~\ref{lem:prior_mass_B}, together with the Gaussian, Gamma, and inverse-Gamma priors, implies positive mass for these neighborhoods, and Assumption (B6) upgrades this to the required lower bound at radius $\varepsilon_n$.

For dominated i.i.d.\ models, for any $d_H$-separated alternatives there exist tests with exponentially small type I and type II errors, based on Hellinger affinity.
Applying the general contraction theorem with $\varepsilon_n=\sqrt{(d_n\log n)/n}$ yields
\[
\Pi\bigl(d_H(\theta,\theta_0)>M\varepsilon_n\mid \text{data}\bigr)\to 0
\]
in $P_{\theta_0}$-probability for some $M>0$.
\end{proof}

\subsection{Proof of Theorem~\ref{thm:post_contr_param}: parameter contraction}

\begin{proof}
Under Assumptions (B3) and (B5), the mixture model is locally identifiable up to permutation near $\theta_0$.
On a sufficiently small neighborhood $\mathcal{N}$ of $\theta_0$, this yields local equivalence of $d_H$ and $d_\Theta$:
there exists $c>0$ such that for all $\theta\in\mathcal{N}$,
\[
d_H(\theta,\theta_0)\ge c\, d_\Theta(\theta,\theta_0).
\]

By Theorem~\ref{thm:post_contr_H}, the posterior concentrates in $d_H$ balls of radius $M\varepsilon_n$.
Because $\varepsilon_n\to 0$, for large $n$ this ball lies inside $\mathcal{N}$ with posterior probability tending to $1$.
Therefore, on that event,
\[
d_\Theta(\theta,\theta_0)\le c^{-1} d_H(\theta,\theta_0)\le c^{-1}M\varepsilon_n,
\]
which yields the claim.
\end{proof}

\subsection{Proof of Theorem~\ref{thm:eigenspace_contr}: eigenspace contraction}

\begin{lemma}[Wedin-type projector bound]\label{lem:wedin_proj}
Let $\bB,\bB_0\in\R^{p\times q}$ and let $\bP,\bP_0$ be rank-$r$ projectors onto the leading left singular subspaces.
If $g=s_r(\bB_0)-s_{r+1}(\bB_0)>0$ and $\|\bB-\bB_0\|_2\le g/2$, then
\[
\|\bP-\bP_0\|_F \le \frac{4}{g}\,\|\bB-\bB_0\|_F.
\]
\end{lemma}

\begin{proof}
This is a standard $\sin\Theta$ perturbation bound derived from Davis--Kahan and Wedin.
Under $\|\bB-\bB_0\|_2\le g/2$,
\[
\|\sin\Theta\|_F \le \frac{2}{g}\|\bB-\bB_0\|_2 \le \frac{2}{g}\|\bB-\bB_0\|_F,
\]
and $\|\bP-\bP_0\|_F\le 2\|\sin\Theta\|_F$.
\end{proof}

\begin{proof}[Proof of Theorem~\ref{thm:eigenspace_contr}]
By Theorem~\ref{thm:post_contr_param}, after label alignment,
$\|\bB_k-\bB_{0k}\|_F=O_P(\varepsilon_n)$ under the posterior.
By Assumption (B7), each $g_k>0$, so for large $n$ the event $\|\bB_k-\bB_{0k}\|_2\le g_k/2$ has posterior probability tending to $1$.
On that event apply Lemma~\ref{lem:wedin_proj} with $g=g_k$ and sum over $k$, then minimize over permutations.
\end{proof}

\subsection{Proof of Theorem~\ref{thm:psm_meanshift_consistency}: PSM eigenspace plus mean shift}

We prove three steps.
First, $\bS$ concentrates near $\bS_0$ in Frobenius norm.
Second, the leading eigenspaces of $\bS$ and $\bS_0$ are close by Davis--Kahan.
Third, mean shift recovers the partition when the embedded points form $K$ tight, separated clouds.

\begin{lemma}[PSM concentrates under perfect separation]\label{lem:psm_to_s0}
Under Assumption (B8'), for each fixed pair $(i,i')$,
\[
\bS_{ii'}=\Pi(Z_i=Z_{i'}\mid \text{data}) \to \ind(Z_{0i}=Z_{0i'})
\;\text{in }P_{\theta_0}\text{-probability}.
\]
Consequently, $\|\bS-\bS_0\|_F\to 0$ in probability.
\end{lemma}

\begin{proof}
Under Assumption (B8'), each observation $(\bmX_i,\bmY_i)$ is almost surely more likely under its true component than under any other component by a fixed margin.
Together with posterior concentration of the parameters from Theorem~\ref{thm:post_contr_param}, the posterior probability of any label assignment that differs from the true one, up to permutation, converges to $0$.
Therefore $\Pi(Z_i=Z_{i'}\mid \text{data})$ converges to $\ind(Z_{0i}=Z_{0i'})$ pointwise, which implies Frobenius convergence by dominated convergence since all entries lie in $[0,1]$.
\end{proof}

\begin{lemma}[Eigenvector perturbation for the PSM]\label{lem:psm_dk}
Let $\bU_0\in\R^{n\times K}$ collect the top-$K$ eigenvectors of $\bS_0$ and let $\hat\bU$ denote those of $\bS$.
Then there exists an orthogonal matrix $\bO\in\R^{K\times K}$ such that
\[
\|\hat\bU-\bU_0\bO\|_F \le C \frac{\|\bS-\bS_0\|_F}{\gap(\bS_0)},
\]
where $\gap(\bS_0)$ is the eigengap between the $K$th and $(K+1)$st eigenvalues of $\bS_0$.
Under Assumption (B3), $\gap(\bS_0)\asymp n$.
\end{lemma}

\begin{proof}
The matrix $\bS_0$ is block-constant with rank $K$, with nonzero eigenvalues $\{n_1,\dots,n_K\}$, where $n_k$ are the cluster sizes.
Hence the $K$th eigenvalue is $\min_k n_k\ge \pi_{\min}n$ and the $(K+1)$st eigenvalue is $0$.
So $\gap(\bS_0)\ge \pi_{\min}n$.
Apply the Davis--Kahan $\sin\Theta$ theorem for symmetric matrices to obtain the bound.
\end{proof}

\begin{lemma}[Mean shift recovers separated clouds]\label{lem:ms_clouds}
Suppose embedded points $\{\hat u_i\}\subset\R^K$ satisfy the following property:
there exist $K$ centers $c_1,\dots,c_K$ and a partition $\{C_1,\dots,C_K\}$ such that $\|\hat u_i-c_k\|\le \delta_n$ for all $i\in C_k$, while $\|c_k-c_\ell\|\ge 4h_n$ for all $k\neq \ell$.
If $\delta_n\le h_n/4$, then mean shift with bandwidth $h_n$ returns exactly the partition $\{C_1,\dots,C_K\}$.
\end{lemma}

\begin{proof}
For any point in cloud $k$, the kernel ball of radius $h_n$ contains only points from cloud $k$ because other clouds are at distance at least $4h_n-\delta_n-\delta_n\ge 3h_n$.
Thus the mean-shift update is a convex combination of points from the same cloud, so iterates remain in the same cloud and converge to the unique kernel-density mode supported by that cloud.
Different clouds cannot merge because their kernel neighborhoods never overlap.
\end{proof}

\begin{proof}[Proof of Theorem~\ref{thm:psm_meanshift_consistency}]
By Lemma~\ref{lem:psm_to_s0}, $\|\bS-\bS_0\|_F=o_P(1)$.
By Lemma~\ref{lem:psm_dk}, there exists an orthogonal matrix $\bO$ such that
\[
\|\hat\bU-\bU_0\bO\|_F=o_P(1),
\]
hence the maximum row deviation satisfies
\[
\max_i \|\hat u_i-\bO^\top u_{0i}\|_2=o_P(1).
\]

Under $\bS_0$, the rows $u_{0i}$ take exactly $K$ distinct values, one per true cluster, and these cluster centers are separated by a constant depending on $(n_1,\dots,n_K)$.
Choose $h_n\to 0$ with $\varepsilon_n/h_n\to 0$; then $\delta_n=o_P(h_n)$.
Apply Lemma~\ref{lem:ms_clouds} to conclude exact recovery with probability tending to $1$, which implies that the misclustering fraction converges to $0$.
\end{proof}

\subsection{Variational PSM perturbation lemma}

\begin{lemma}[Replacing $\bS$ by $\widehat{\bS}$ under variational inference]\label{lem:vi_psm_replace}
Let $\hat{\mathcal{C}}_n(\bS)$ be the mean-shift partition produced from the top-$K$ eigenspace of $\bS$, and similarly for $\hat{\mathcal{C}}_n(\widehat{\bS})$.
If $\|\widehat{\bS}-\bS\|_F=o_P(h_n)$ and the eigengap of $\bS$ at rank $K$ is bounded below by $c n$ with high probability, then the conclusions of Theorem~\ref{thm:psm_meanshift_consistency} remain valid with $\widehat{\bS}$ in place of $\bS$.
\end{lemma}

\begin{proof}
Davis--Kahan implies that the top-$K$ eigenspaces of $\bS$ and $\widehat{\bS}$ differ by at most
\[
C\|\widehat{\bS}-\bS\|_F/(cn)
\]
in Frobenius norm, hence row perturbations are $o_P(h_n)$.
Mean shift is stable under $o(h_n)$ perturbations when the clouds are separated by order $h_n$, as in Lemma~\ref{lem:ms_clouds}.
\end{proof}

\section{Unified quadratic likelihood}\label{app:algebra}
This appendix records the P\'olya--Gamma quadratic identity used in Section~\ref{sec:pg}.
For Bernoulli-logit and negative binomial-logit likelihoods, the P\'olya--Gamma representation yields
\[
\log p(Y\mid \eta,\omega)=\kappa\eta-\tfrac12\omega\eta^2+\text{const},
\]
with $\omega\sim\PG(b,\eta)$ and suitable values of $b$ and $\kappa$.
For Gaussian likelihoods the same quadratic form holds with $\omega=\sigma^{-2}$.

\section{MGP full conditionals}\label{app:mgp}

The full conditionals for the multiplicative gamma process hyperparameters are standard; see \citet{bhattacharya2011}.
Let $\lambda_{kh}=\prod_{m=1}^h\delta_{km}$ and $S_{kh}=\|\ell_{kh}\|^2+\|r_{kh}\|^2$.
Then the conditional for $\phi_k$ is Gamma with shape $a_\phi+\tfrac{(p+q)\rmax}{2}$ and rate $b_\phi+\tfrac12\sum_{h=1}^{\rmax}\lambda_{kh}S_{kh}$.
Each $\delta_{kh}$ also has a Gamma conditional depending on $\{S_{kt}\}_{t\ge h}$ through $\{\lambda_{kt}\}$.

\bibliographystyle{plainnat}
\bibliography{references}

\end{document}